\documentclass[envcountsame]{llncs}

\title{State Identification for Labeled Transition Systems with Inputs and Outputs\thanks{Funded by the Netherlands Organisation of Scientific Research (NWO) under project 13859: Supersizing Model-Based testing (SUMBAT).}}

\author{Petra van den Bos  \and Frits Vaandrager}
\institute{Institute for Computing and Information Sciences, \\ Radboud University, Nijmegen, the~Netherlands\\\texttt{$\{$petra, f.vaandrager$\}$@cs.ru.nl}}
\usepackage{amssymb}
\usepackage[linesnumbered,noline]{algorithm2e} 
\usepackage{enumerate}
\usepackage[fleqn]{amsmath}
\usepackage{mathtools}
\usepackage{tikz}
\usetikzlibrary{patterns, intersections, positioning, shapes.geometric, calc,
automata, arrows}
\usepackage{placeins}
\usepackage{wrapfig}
\usepackage{cancel}
\usepackage{pdfpages}
\usepackage{float}
\usepackage{pdflscape}
\usepackage{hyperref}
\usepackage{wrapfig}
\SetAlCapSkip{2mm} 

\newif\iflong
\longtrue

\newif\ifhideproofs
\hideproofsfalse

\ifhideproofs
  \usepackage{environ}
  \NewEnviron{hidep}{}

\fi

\usepackage{etoolbox}
\AtBeginEnvironment{align}{\setcounter{equation}{0}}

\usepackage[inline]{enumitem}

\newcommand{\algorithmref}[1]{Algorithm~\ref{#1}}
\newcommand{\figureref}[1]{Figure~\ref{#1}}
\newcommand{\exampleref}[1]{Example~\ref{#1}}
\newcommand{\definitionref}[1]{Definition~\ref{#1}}
\newcommand{\lemmaref}[1]{Lemma~\ref{#1}}
\newcommand{\corollaryref}[1]{Corollary~\ref{#1}}

\newcommand{\tableref}[1]{Table~\ref{#1}}

\newcommand{\after}[2]{#1\text{ }\textit{after}\text{ }#2}
\newcommand{\before}[2]{#1\text{ }\textit{before}\text{ }#2}

\newcommand{\out}[1]{\textit{out}(#1)}
\newcommand{\outm}[1]{\textit{out}(#1)}
\newcommand{\outsub}[2]{\textit{out}_{#1}(#2)}
\newcommand{\inp}[1]{\textit{in}(#1)}
\newcommand{\inpm}[1]{\textit{in}(#1)}

\newcommand{\traces}[1]{\textit{traces}(#1)}
\newcommand{\pre}{\mathit{Pre}}
\newcommand{\post}{\mathit{Post}}
\newcommand{\enabled}[2]{\textit{enabled}(#1,#2)}

\newcommand{\Obs}[1]{\textit{Obs}(#1)}
\newcommand{\Nil}{\mathbf{0}}

\newcommand{\Ev}{E_\mathit{CCS}}
\newcommand{\Tv}{T_\mathit{CCS}}

\newcommand{\smallerfont}{\fontsize{7pt}{0}\selectfont}

\newcommand{\compatible}{\mathrel{\Diamond}}

\newcommand{\leaves}[1]{\textit{leaves}(#1)}
\newcommand{\nonleaves}[1]{\textit{internal}(#1)}

\newcommand{\bigO}{\mathcal{O}}

\SetKw{Continue}{continue}
\SetKw{Break}{break}

\let\oldnl\nl
\newcommand{\nonl}{\renewcommand{\nl}{\let\nl\oldnl}}

\usepackage{enumitem}
\newcommand*{\descriptiontext}[1]{(#1)}
\setlist[description]{font=\textbullet\normalfont\space\descriptiontext} 

\begin{document}
 \maketitle
\begin{abstract}
For Finite State Machines (FSMs) a rich testing theory has been developed to discover
aspects of their behavior and ensure their correct functioning.
Although this theory is widely used, e.g., to check conformance of protocol implementations,
its applicability is limited by restrictions of the FSM framework: the fact that inputs
and outputs alternate in an FSM, and outputs are fully determined by the previous input and state.
Labeled Transition Systems with inputs and outputs (LTSs), as studied in ioco testing theory, provide a richer
framework for testing component oriented systems, but lack the algorithms for test generation from FSM theory.

In this article, we 
propose an algorithm for the fundamental problem of \emph{state identification} during testing of LTSs. Our algorithm is a direct generalization of the well-known algorithm for computing adaptive distinguishing sequences for FSMs proposed by Lee \& Yannakakis.
Our algorithm has to deal with so-called \emph{compatible} states, states that cannot be distinguished in case of
an adversarial system-under-test.
Analogous to the result of Lee \& Yannakakis, we prove that if an (adaptive) test exists that distinguishes all pairs of incompatible states of an LTS, our algorithm will find one.
In practice, such adaptive tests typically do not exist. However, in experiments with an implementation of our algorithm on an
industrial benchmark, we find that tests produced by our algorithm still distinguish more than 99\% of the incompatible state pairs.
\end{abstract}

\section{Introduction}

Starting with Moore's famous 1956 paper \cite{moore1956}, a rich theory of testing finite-state machines (FSMs) has been developed to discover
aspects of their behavior and ensure their correct functioning; see e.g.\ \cite{LeeY96} for a survey.
One of the classical testing problems is \emph{state identification}: given some FSM, determine in which state it was initialized, by providing inputs and observing outputs.

Various forms of \emph{distinguishing sequences} were proposed, ranging from sets of sequences to single sequences solving the problem. 
Moreover, when combined with state access sequences, so called $n$-complete test suites can be constructed \cite{fsmtestingsurvey}.
%
The challenge in using $n$-complete test suites is to keep their size as small as possible. 
Using a single (adaptive) sequence for state identification \cite{leeyannakakis}, helps to reach this objective.
If such a single sequence does not exist, then a distinguishing sequence distinguishing most states may be supplemented with some additional distinguishing sequences that distinguish the remaining states \cite{thesisjoshua}.

Although state identification algorithms for FSMs have been widely used, e.g., to check conformance of protocol implementations,
their applicability is limited by the expressivity of the FSM framework.
In FSMs, inputs and outputs strictly alternate, outputs are fully determined by the previous input and state, and inputs must be enabled in every state.
Labeled Transition Systems with inputs and outputs (LTSs), 
as studied in ioco testing theory \cite{tretmans}, provide a richer
framework for testing component oriented systems:
transitions are labeled by either an input or an output, allowing any combination of inputs and outputs, multiple outputs may be starting from the same state, allowing (observable) output nondeterminism, and states do not need to have transitions for all inputs, allowing partiality.
However, LTSs lack the algorithms for test generation from FSM theory.
Although progress has been made in defining and constructing $n$-complete test suites for LTSs \cite{ncompletejournal}, an algorithm to solve the state identification problem as in \cite{leeyannakakis}, and hence to provide slim $n$-complete test suites, is missing.

Therefore we generalize the construction algorithms for adaptive distinguishing sequences, as given in \cite{leeyannakakis}.
As in \cite{ncompletejournal}, we have to face the problem of compatible states \cite{detimpnondetspec,iotsioco}, which does not occur for FSMs.
States are \emph{compatible} when they cannot be distinguished in case of
an adversarial system-under-test, e.g. when two states have a transition for the same output to the same state.
As it is easy to construct LTSs with compatible states, we made sure our algorithms can deal with such LTSs: they accept LTSs with compatible states, but they `work around' them, dealing with all incompatible states.

The outline of the paper is as follows. We first introduce graphs, LTSs, and some syntax for denoting trees.
Then we elaborate on compatibility and the related concept of validity.
Furthermore, we introduce test cases, and define when they distinguish states of an LTS.
After that we define a data structure called \emph{splitting graph}, present an algorithm that constructs a splitting graph for a given LTS, and another algorithm that extracts a test case from a splitting graph.
We show that, unlike for FSMs, the splitting graph may have an exponential number of nodes. However, this is worst case behaviour, as our experiments on an industrial case study will show.
Analogous to FSMs, it may not be possible to distinguish all states of an LTS with a single test case. 
Our experiments show that this is typically the case in practice, but nevertheless more than 99\% of the incompatible state pairs are distinguished by the constructed test case. 
Following \cite{leeyannakakis}, we show that our algorithms constructs a test case distinguishing all incompatible state pairs, if it exists.

\paragraph{Related work}
There are (at least) three ortogonal ways in which the classical FSM (or Mealy machine) model can be generalized.

A first generalization is to add nondeterminism. Whereas an FSM has exactly one outgoing transition for each state $q$
and input $i$, a \emph{nonderministic FSM} allows for more than one transition.
Alur, Courcoubetis \& Yannakakis \cite{distnondetprob} propose an algorithm to generate adaptive distinguishing sequences for nondeterministic FSMs, using (overlapping) subsets  of states, similar to our algorithm. However, their sequences only distinguish pairs of states, and are not designed to distinguish more states at the same time.
In between FSMs and nondeterministic FSMs we find the \emph{observable} FSMs, which have at most one outgoing transition  for each state $q$, input $i$ and output $o$; one may use a determinization construction to convert any nondeterministic FSM into an
observable one. The LTSs that we consider have observable nondeterminism.

A second generalization of FSMs is to relax the requirement that each input is enabled in each state.  In a \emph{partial FSM}, states
do not necessarily have outgoing transitions for every state and every input.
Petrenko \& Yevtushenko \cite{nfsmreduction} derive complete test suites for partial, observable FSMs, 
which is the closest to the automata model that we study in this paper.
Their test generation is based on (adaptive) state counting \cite{statecounting}, which is a trace search-based method which recognizes when states are distinguished, but does not provide a constructive way to build a test that distinguishes (many) states at once.
Yannakakis \& Lee \cite{LYfaultdetection} present a randomized algorithm which generates, with high probability, checking sequences, i.e., $n$-complete test suites consisting of a single sequence. This approach is also applicable to partial FSMs, as opposed to the adaptive distinguishing sequence construction algorithms of \cite{leeyannakakis}, which apply to plain FSMs.

A third generalization of FSMs is to relax the requirement that inputs and outputs alternate. In our LTS, inputs and outputs
may occur in arbitrary order. 
Bensalem, Krichen \& Tripakis \cite{detstateidfc} give an algorithm for extracting adaptive distinguishing sequences for all states of a given LTS, by translating back and forth between a corresponding Mealy machine. This translation is only possible, if all states of the LTS have at most one outgoing output transition. 
Van den Bos, Janssen \& Moerman \cite{ncompletejournal} do not need such a restriction. 
They propose an algorithm that generates an adaptive distinguishing sequence for all pairs of incompatible states. In this paper, we generalize the result of \cite{ncompletejournal} to distinguish more states at the same time.

\iflong
\noindent
\fi

\section{Preliminaries}
\label{sec:prelim}
We write $f : X \rightharpoonup Y$ to denote that $f$ is a partial function from $X$ to
$Y$.  We write $f(x) \downarrow$ to mean $\exists y : f(x)=y$, i.e., the result
is defined, and $f(x) \uparrow$ if the result is undefined.
We often identify a partial function $f$ with the set of pairs
$\{ (x,y) \in X \times Y \mid f(x)=y \}$.

If $\Sigma$ is a set of symbols then $\Sigma^{\ast}$ denotes the set of all
finite words over $\Sigma$. The empty word is denoted by $\epsilon$, the word 
consisting of symbol $a \in \Sigma$ is denoted $a$, and concatenation of words is denoted
by juxtaposition.

Throughout this article, we use standard notations and terminology related to finite directed graphs (digraphs) and finite directed acyclic graphs (DAGs), as for instance defined in \cite{CLRS3rd,BK08}.
If $G = (V, E)$ is a digraph and $v \in V$, then we let $\post_G(v)$, or
briefly $\post(v)$, denote the set of direct successors of $v$, that is,
$\post(v) = \{ w \in V \mid (v,w) \in E \}$. 
Similarly, $\pre_G(v)$, or briefly $\pre(v)$, denotes the set of direct predecessors
of $v$, that is, $\pre(v) = \{ w \in V \mid (w,v) \in E \}$.
Vertex $v$ is called a \emph{root} if $\pre(v) = \emptyset$, a \emph{leaf} if
$\post(v) = \emptyset$, and \emph{internal} if $\post(v) \neq \emptyset$.
We write $\leaves{G} = \{v \in V \mid \post(v) = \emptyset\}$, and $\nonleaves{G} = V \setminus \leaves{G}$.

The automata considered this paper are deterministic, finite labeled transition systems with transitions that are labeled by inputs or outputs.
Since a single state may have outgoing transitions labeled with different outputs, and since outputs are not controllable, the behavior of our automata is nondeterministic: in general, for a given sequence of inputs, the resulting sequence of outputs is not uniquely determined.  Nevertheless, our automata are deterministic in the sense of classical automata theory: for any observed sequence of inputs and outputs the resulting state is uniquely determined.  We say that our automata have \emph{observable nondeterminism}.

Because the inputs and outputs will be fixed throughout this article, we fix $I$ and $O$ as nonempty, disjoint, finite sets of input and output labels, respectively, and write $L = I \cup  O$.
We will use $a, b$ to denote input labels, $x, y, z$ to denote output labels, and
$\mu$ for labels that are either inputs or outputs.

\begin{definition} \label{def:automata}
 An \emph{automaton (with inputs and outputs)} is a triple $A = (Q,T,q_0)$ with 
 $Q$ a finite set of \emph{states}, $T : Q \times L \rightharpoonup Q$
 a \emph{transition function}, and $q_0 \in Q$ the \emph{initial state}.
 We associate a digraph to $A$ as follows
 \begin{eqnarray*}
 \mathit{digraph}(A) & = & (Q, \{ (q, q') \mid \exists \mu \in L : T(q,\mu) = q'\}).
 \end{eqnarray*}
Concepts and notations for $\mathit{digraph}(A)$ extend to $A$. Thus we say,
for instance, that automaton $A$ is acyclic when $\mathit{digraph}(A)$ is acyclic, and we write $\post(q)$ for the set of direct successors of a state $q$.
For $A = (Q,T,q_0)$ and $q \in Q$ we write $A / q$ for $(Q,T,q)$, that is,
the automaton obtained from $A$ by replacing the initial state by $q$.
\end{definition}

 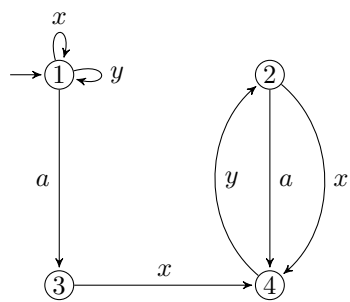
\begin{wrapfigure}{r}{0.4\textwidth}
	\centering
	\begin{tikzpicture}[shorten >=1pt,node distance=2.8cm,>=stealth']
	\tikzstyle{every state}=[draw=black,text=black,inner sep=1pt,minimum
	size=10pt,initial text=]
	\node[state,initial,initial where=left] (1) {1};
	\node[state] (2) [right of=1] {2};
	\node[state] (3) [below of=1] {3};
	\node[state] (4) [below of=2] {4};
	\path[->]
	(1) edge [loop above] node {$x$} (1)
	(1) edge [loop right] node {$y$} (1)
	(1) edge node [left] {$a$} (3)
	(2) edge node [right] {$a$} (4)
	(2) edge [bend left=50] node [right] {$x$} (4)
	(3) edge node [above] {$x$} (4)
	(4) edge [bend left=50] node [right] {$y$} (2)
	;
	\end{tikzpicture}
	\caption{Running example}
	\label{fig:spec}
			\vspace{2mm}
\end{wrapfigure}

\figureref{fig:spec} shows an example automaton.
Below, we recall the definitions of some basic operations on (sets of) automata states.
Operations $\mathit{in}$, $\mathit{out}$ and $\mathit{init}$ retrieve all the inputs, outputs, or labels enabled in a state, respectively.
To every set of states $P$ and every sequence of labels $\sigma$ we can associate three sets of states:
$\after{P}{\sigma}$, $\before{P}{\sigma}$, and $\enabled{P}{\sigma}$. 
The set $\after{P}{\sigma}$ comprises all states that can be reached starting from a state of $P$ via a path with trace $\sigma$,
whereas the set $\before{P}{\sigma}$ consists of all the states from where it is possible to reach a state in $P$ via a trace in $\sigma$,
and $\enabled{P}{\sigma}$ consists of all states in $P$ from where a path with trace $\sigma$ is possible.
The \textit{traces} operation provides the sequences of labels that can be observed from one or more of the states. 
We use a subscript if confusion may arise due to the use of several automata in the same context, e.g. $\outsub{A}{q}$ denotes the enabled outputs of $q$ in automaton $A$.

 \begin{definition} \label{def:operations}
 Let $A = (Q,T,q_0)$ be an automaton, 
 $q \in Q$,  $\mu \in L$ and $\sigma \in L^*$. Then we define: 

 \vspace{-1em}
 \noindent
 \begin{minipage}[t]{0.45\textwidth}
\small
 \begin{align*}
 & \inpm{q} = \{a \in I \mid T(q,a) \downarrow\}\\
 & \outm{q} = \{x \in O \mid T(q,x) \downarrow\}\\
  &\after{q}{\epsilon} = \{q\}\\
  &\after{q}{\mu\sigma} =
  \begin{cases}
     \after{T(q,\mu)}{\sigma} & \text{if } T(q,\mu) \downarrow\\
     \emptyset & \text{otherwise}
   \end{cases}\\
 \end{align*}
\end{minipage}
\quad
 \begin{minipage}[t]{0.53\textwidth}
\small
 \begin{align*}
 & \enabled{q}{\sigma} = \begin{cases}
 \emptyset & \text{if } \after{q}{\sigma} = \emptyset\\
 \{ q \} & \text{otherwise}
 \end{cases}\\
  &\before{q}{\sigma} = \{ q' \in Q \mid q \in \after{q'}{\sigma} \}\\
 &\after{A}{\sigma} = \after{q_0}{\sigma}\\   
 & \traces{q} = \{\rho \in L^* \mid \after{q}{\rho} \neq \emptyset\}
 \end{align*}
\end{minipage}

\noindent 
Definitions are lifted to sets of states by pointwise extension. Thus, for
$P \subseteq Q$,
$\inpm{P} = \bigcup_{p \in P}\inpm{p}$,
$\after{P}{\sigma} = \bigcup_{p \in P} \after{p}{\sigma}$, etc.
We sometimes write the automaton, instead of the singleton set containing the initial state.
\end{definition}

We find it convenient to use a fragment of Milner's Calculus of Communicating Systems \cite{milner} as syntax for denoting acyclic automata. In particular, its recursive definition will allow us to incrementally construct test cases in Sections~\ref{sec:algorithms} and~\ref{sec:distgraph}.

\begin{definition} \label{def:CCS}
The set of expressions $\Ev$  is defined by the BNF grammar 
\begin{eqnarray*}
F & := & \Nil \mid F + F \mid \mu.F
\end{eqnarray*}
The set $\Tv \subseteq \Ev \times L \times \Ev$ is the smallest set of triples such that,
for all $\mu \in L$ and $F, F', G \in \Ev$,
\begin{enumerate}
\item 
$(\mu . F, \mu, F) \in \Tv$
\item
If $(F, \mu, G) \in \Tv$ then $(F+F', \mu, G) \in \Tv$
\item
If $(F, \mu, G) \in \Tv$ then $(F'+F, \mu, G) \in \Tv$
\end{enumerate}
An expression $F \in \Ev$ is \emph{deterministic} iff, for all subexpressions
$G$ of $F$,
\begin{eqnarray*}
(G,\mu, G') \in \Tv \wedge (G,\mu, G'') \in \Tv & \Rightarrow & G'=G''
\end{eqnarray*}
To each deterministic expression $F \in \Ev$ we associate an automaton $A_F = (Q, T, F)$, where $Q$ is the set of subexpressions of $F$, and
transition function $T$ is defined by
\begin{eqnarray*}
T(G,\mu) & = & \left\{ \begin{array}{ll}
G' & \mbox{if } (G,\mu,G') \in \Tv\\
\mbox{undefined} & \mbox{otherwise}
\end{array}\right.
\end{eqnarray*}
\end{definition}

\begin{example} \label{exmp:CCS}
 The CCS expression $a.(x.\Nil + y.\Nil)$ has subexpressions $a.(x.\Nil + y.\Nil)$, $x.\Nil + y.\Nil$, $x.\Nil$, $y.\Nil$, and $\Nil$. These are the states of its associated automaton. The automaton's transition relation is: $\{ (a.(x.\Nil + y.\Nil), a , x.\Nil + y.\Nil), (x.\Nil + y.\Nil, x, \Nil), (x.\Nil + y.\Nil, y,\Nil) (x.\Nil, x, \Nil), (y.\Nil, y,\Nil) \}$. Note that states $x.\Nil$ and $y.\Nil$ are not reachable from initial state $a.(x.\Nil + y.\Nil)$.
\end{example}

Suspension automata are automata with the additional property that in each state at least one output label is enabled. 
We note that this requirement can be easily enforced by adding a self-loop for an additional output label, that denotes `no-output' or \emph{quiescence} \cite{tretmans}, in each state that has no output transition.
We note that our definition of suspension automata, which is taken from \cite{ncompletejournal}, is more general than the one from \cite{tretmans,willemse}, since we only require states to be non-blocking, while suspension automata from \cite{tretmans,willemse} adhere to some additional properties associated to this special quiescence output. 

\begin{definition} \label{def:sa}
  Let $A = (Q,T,q_0)$ be an automaton.
   We call a state $q \in Q$ \emph{blocking} if  $\out{q} = \emptyset$, and call
   $A$ \emph{non-blocking} if none of its states is blocking.
  A non-blocking automaton is also called a \emph{suspension automaton}.
\end{definition}

We will use suspension automata as the specifications to derive test cases from. \figureref{fig:spec} shows a suspension automaton.
Plain automata are sometimes used as an intermediate structure to do computations, and test cases will be acyclic automata adhering to some additional properties.

\section{Validity and Compatibility} \label{sec:valcomp}
In this section, we recall the definitions of the related notions of validity and
compatibility \cite{ncompletejournal}.
We first give an efficient algorithm for computing valid states.
After that, we show how the relation between validity and compatibility can be used to efficiently compute all pairs of compatible states occurring in a suspension automaton. We will need this last relation when constructing test cases to distinguish incompatible states.

\subsection{Validity}
\label{sec:validity}
We consider the following 2-player concurrent game, which is a minor variant of
reachability games studied e.g., in \cite{Mazala2002,mbtgames}.
Two players, the tester and the System Under Test (SUT), play on a state space consisting of an
automaton $A = (Q,T,q_0)$.
At any point during the game there is a \emph{current state}, which is $q_0$ initially.
To advance the game, both the tester and the SUT choose an action from the current state $q$:
\begin{itemize}
	\item 
	The tester chooses either an input from $\inpm{q}$, or the special action $\theta \not\in L$. By choosing $\theta$, the tester indicates that she performs no
	input and allows the SUT to execute any output he wishes.
	\item 
	The SUT chooses an output from $\outm{q}$, or $\theta$ if no output is possible.
\end{itemize}
The game moves to a next state according to the following rule
(this is the input-eager assumption from \cite{mbtgames}):
If the tester chooses an enabled input $a$ this will be executed, i.e.,
the current state changes to $T(q,a)$;
if the SUT chooses an enabled output $x$ this will only be executed when the tester
has chosen $\theta$, in this case the current state changes to $T(q,x)$; when both players choose $\theta$, the game terminates.
The tester wins the game if she reaches a blocking state, and the SUT wins if he has a strategy that ensures that the tester will never win.
A (memoryless) strategy for the tester is a function ${\mathit move}: Q \rightarrow I \cup \{ \theta \}$. We say a strategy is \emph{winning} if the tester will
always win the game (within a finite number of moves) when selecting actions according to this strategy, no matter
which actions the SUT takes.
Following Bene{\v{s}} et al \cite{mergesas} and Van den Bos et al \cite{ncompletejournal}, we call states for which the tester
has a winning strategy \emph{invalid}, and the remaining states in $Q$ \emph{valid}.
The sets of valid and invalid states are characterized by the following
lemma (cf Proposition 2.18 of \cite{Mazala2002}):
\begin{lemma}
	\label{lemma game}
Let  $A = (Q,T,q_0)$ be an automaton.
\begin{enumerate}
\item 
The set of invalid states of $A$ is the smallest set $P \subseteq Q$ such that $q \in P$ if
\begin{eqnarray*}
	\exists a \in \inp{q}: T(q,a) \in P & \mbox{ or} &
	\forall x \in \outm{q} : T(q,x) \in P.
\end{eqnarray*}
\item 
The set of valid states of $A$
is the largest set $P \subseteq Q$ such that $q \in P$ implies
\begin{eqnarray*}
\forall a \in \inp{q}: T(q,a) \in P & \mbox{ and} &
\exists x \in \outm{q} : T(q,x) \in P.
\end{eqnarray*}
\end{enumerate}
\end{lemma}


Based on \lemmaref{lemma game}(1), \algorithmref{blockingalg} computes the set of invalid states of an automaton $A$ and, for each invalid state $q$, the first move $\mathit{move}(q)$ of a winning strategy for the tester, as well as the maximum number $\mathit{level}(q)$ of moves required to win the game.
\begin{algorithm}[!ht]
	\caption{Computing the invalid states.}
	\label{blockingalg}
	\SetKwInOut{Input}{Input}
	\Input{An automaton $A = (Q, T, q_0)$.}
	\SetKwInOut{Output}{Output}
	\Output{The subset $P \subseteq Q$ of invalid states and, for each state $q \in P$, the first move $\mathit{move}(q)$ from a winning stragegy for the tester and the maximum number $\mathit{level}(q)$ of moves required to win.}
	\SetKwFunction{func}{ComputeWinningTester}
	\SetKwProg{myalg}{Function}{:}{}
	\myalg{\func($Q, T, q_0$)}{
		$W := \emptyset$ \tcp*{winning states for tester that need processing}
		\ForEach{$q \in Q$}{		
		$\mathit{count}(q) := ~ \mid \outm{q} \mid$\;
		$\mathit{incomingtransitions}(q) :=$ set of incoming transitions of $q$\;
		\If{$\mathit{count}(q)=0$}{
			$W := W \cup \{ q \}$  \tcp*{state $q$ is invalid}
			$\mathit{move}(q) := \theta$\;
			$\mathit{level}(q) := 0$
		}
	}
		$P := \emptyset$ \tcp*{winning states for tester that have been processed}
       \While{$W \neq \emptyset$}
       {
       	$p := \mbox{any element from } W$\;
       	\ForEach{$(q, \mu, p) \in \mathit{incomingtransitions}(p)$}
       	{
       	\If{$q \not\in P \cup W$}{
       		\eIf{$\mu \in I$}{	
       		$W := W \cup \{ q \}$ \tcp*{state $q$ has input to winning state}
       		$\mathit{move}(q) := \mu$\;	
       		$\mathit{level}(q) := \mathit{level}(p)+1$
       	}
       	{
       			$\mathit{count}(q) :=  \mathit{count}(q)-1$\;
       			\If{$\mathit{count}(q)=0$}{
       				$W := W \cup \{ q \}$ \tcp*{all outputs $q$ to winning states}
       				$\mathit{move}(q) := \theta$\;
       				$\mathit{level}(q) := 1+ \max_{x \in\outm{q}}  \mathit{level}(T(q,x))$
       			}
       		}
       	}
       	}
       $W := W \setminus \{ p \}$\; 
       $P := P \cup \{ p \}$
       }	
		\Return{\emph{set} $P$, \emph{function} $\mathit{move}$, \emph{and function} $\mathit{level}$}\;
	}
\end{algorithm}
Algorithm~\ref{blockingalg} is a minor
variation of the classical algorithm for computing attractor sets and traps in 2-player concurrent games \cite{Mazala2002}
and the procedure described by Bene{\v{s}} et al \cite{mergesas},
which takes as input an automaton, of which each state $q$ has $\inp{q} = L_I$, and prunes away 
invalid states.
Key invariants of the while-loop of lines 13-33 are that
states in $W \cup P$ are invalid, and
for  $q \in Q \setminus (P \cup W)$, $\mathit{count}(q)$ gives
the number of output transitions to states in $Q \setminus P$.

Let $n$ be the number of states in $Q$, and $m$ the number of transitions in $T$. We assume, for convenience, that $m \geq n$.
If we use an adjacency-list representation of $A$ and represent the
set of incoming transitions using a linked list, 
then the time complexity of the initialization part (lines 2-11) is $\bigO(m)$.
The while-loop (lines 13-33) visits each transition of $A$ at most twice (in lines 15 and 26) and
performs a constant amount of work.
Thus the time complexity of the while loop is $\bigO(m)$.
This means that the time complexity of Algorithm~\ref{blockingalg} is also $\bigO(m)$.

\iflong
The next lemma states some basic properties of the $\mathit{level}$ function that records the maximum number of moves required to win.

\begin{lemma}
\label{lemma move and level}
Let $A = (Q, T, q_0)$ be an automaton and $P \subseteq Q$ the set of invalid states of $A$. Let $\mathit{move}$ and $\mathit{level}$
be as computed by Algorithm~\ref{blockingalg}. Then, for all $q \in P$ and $a \in I$,
\begin{enumerate}
\item
$\mathit{level}(q) = 0 ~ \Leftrightarrow ~ q$ is blocking, 
\item
$\mathit{level}(q) > 0 \wedge \mathit{move}(q)=a  ~ \Rightarrow ~ T(q,a) \in P \wedge
\mathit{level}(T(q,a)) < \mathit{level}(q)$,
\item
$\mathit{level}(q) > 0 \wedge \mathit{move}(q) = \theta ~ \Rightarrow ~ \forall x \in \outm{q}: \mathit{level}(T(q,x)) < \mathit{level}(q)$.
\end{enumerate}
\end{lemma}
\fi


\subsection{Compatibility}
Two states of a suspension automaton are \emph{compatible} \cite{detimpnondetspec,iotsioco} 
if a tester may not
be able to distinguish them in the presence of an adversarial SUT. For example, if the tester wants to determine whether the SUT behaves according to state 2 or 3 of the suspension automaton of \figureref{fig:spec}, taking output transition $x$ will result in reaching state 4, from both states, but after reaching state 4, it cannot be determined, from which of the two states the $x$ transition was taken. Hence, states 2 and 3 are compatible.

\begin{definition} \label{def:compatible}
Let $(Q,T,q_0)$ be a suspension automaton.
A relation $R \subseteq Q \times Q$ is a \emph{compatibility relation} if for all $(q, q') \in R$ we have
  \begin{eqnarray*}
  & &\forall a \in \inpm{q} \cap \inpm{q'}: (T(q,a), T(q',a)) \in R \text{, and}\\
  & &\exists x \in \outm{q} \cap \outm{q'}: (T(q,x), T(q',x)) \in R
  \end{eqnarray*}
Two states $q, q' \in Q$ are \emph{compatible}, denoted $q \compatible q'$, if there exists a compatibility relation $R$ relating $q$ and $q'$.
Otherwise, the states are \emph{incompatible}, denoted by $q \not\compatible q'$.
For $P \subseteq Q$ a set of states, we write $\Diamond(P)$ to denote that all
states in $P$ are pairwise compatible, i.e., $\forall q, q'\in P : q \compatible q'$.
\end{definition}

We note that the compatibility relation is symmetric and reflexive, but not transitive.
For an elaborate discussion of compatibility, we refer the reader to \cite{ncompletejournal}.
The notions of compatibility and validity can be related using
the following synchronous composition operator:

\begin{definition}
\label{def:composition}
Let	$A_1 = (Q_1,T_1,q^1_0)$ and $A_2 = (Q_2,T_2,q^2_0)$ be automata.
The \emph{synchronous composition} of $A_1$ and $A_2$, notation $A_1 \| A_2$, is the automaton
$A = (Q_1 \times Q_2 ,T ,(q^1_0, q^2_0))$, where transition function $T$ is given by:
\begin{eqnarray*}
T((q_1, q_2), \mu) & = & \left\{ \begin{array}{ll}
	(T_1(q_1, \mu), T_2(q_2, \mu)) & \mbox{if } T(q_1,\mu) \downarrow \mbox{ and } T(q_2,\mu) \downarrow\\
	\mbox{undefined} & \mbox{otherwise}
	\end{array}\right.
\end{eqnarray*}
\end{definition}

The next lemma asserts that
states $q$ and $q'$ are compatible precisely when the pair $(q,q')$
is a valid state of $S$ composed with itself.\footnote{This is a  variation of Lemma 22 from \cite{ncompletejournal}, which is stated for a
slightly different composition operator that involves demonic completions.
Adding demonic completions is useful in the setting of \cite{ncompletejournal},
but not needed for our purposes.}

\begin{lemma}
\label{compatibility reduces to validity}
Let $S = (Q,T,q_0)$ be a suspension automaton with $q, q'\in Q$.
Then $q \compatible q'$ if and only if $(q, q')$ is a valid state of $S \| S$.
\end{lemma}
\begin{proof}
    ($\Leftarrow$)
	Suppose that $(q, q')$ is a valid state of $S \| S$.
	Then, by \lemmaref{lemma game}(1), $(q,q')$ is contained in the largest subset $P$ of
	the states of $S \| S$ that satisfies the conditions of \lemmaref{lemma game}(2).
	Using \definitionref{def:composition}, we infer that, for all $(r, r') \in P$:
	\begin{eqnarray*}
	&&\forall a \in \inpm{r} \cap \inpm{r'}: (T(r,a), T(r',a)) \in P \text{, and}\\
	&&\exists x \in \outm{r} \cap \outm{q'}: (T(r,x), T(r',x)) \in P
	\end{eqnarray*}
	But this means that $P$ is a compatibility relation, and therefore $q \compatible q'$.

	($\Rightarrow$) Suppose that $q \compatible q'$.
	Then, by \definitionref{def:compatible}, there exists a compatibility
	relation $R$ relating $q$ and $q'$.
	Since $R \subseteq Q \times Q$, $R$ is a subset of the set of states
	of $S \| S$.
	By combining Definitions~\ref{def:compatible} and \ref{def:composition},
	we infer that $R$ is the set $P$ from \lemmaref{lemma game}(2).
	This implies that $(q, q')$ is a valid state of $S \| S$. \qed
\end{proof}

\begin{example}
 \figureref{fig:comptransitive} shows the synchronization of the suspension automaton of \figureref{fig:spec}. It has 6 valid states, and in particular it shows that $2 \compatible 3$.
\end{example}

\lemmaref{compatibility reduces to validity} suggests an efficient algorithm
for computing compatibility of states.
Suppose $S$ is a suspension automaton with $n$ states and $m$ transitions,
with $m \geq n$.
Then we may compute composition $S \| S$ in time $\bigO(m(n + \log m))$.
The idea is that we first sort the list of transitions on the value of their
action label, which takes $\bigO(m \log m)$ time.
Next we check for each transition $t= (q, \mu, q')$ what are the possible transitions
that may synchronize with $t$. Since $t$ may only synchronize with $\mu$-transitions,
and since there are at most $n$ $\mu$-transitions (as $S$ is deterministic), we may
compute the list of transitions of the composition in $\bigO(m n)$ time.
Thus, the overall time complexity of computing $S \| S$ is $\bigO(m(n+\log m))$.
The composition $S \| S$ has $n^2$ states and $\bigO(m n)$ transitions.
Next we use Algorithm~\ref{blockingalg} to compute the set of invalid states of $S \| S$,
which requires $\bigO(m n)$ time. Two states $q$ and $q'$ of $S$ are compatible
iff $(q,q')$ is not in this set.
Altogether, we need $\bigO(m(n+ \log m))$ time to compute the 
compatible state pairs.

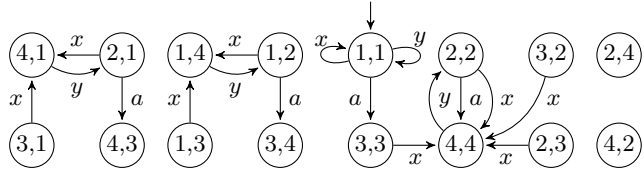
\begin{figure}
\centering
 \begin{tikzpicture}[shorten >=1pt,node distance=1.2cm,>=stealth']
	\tikzstyle{every state}=[draw=black,text=black,inner sep=1pt,minimum
	size=10pt,initial text=]
	\node[state,initial,initial where=above] (1) {1,1};
	\node[state] (2) [right of=1] {2,2};
	\node[state] (3) [below of=1] {3,3};
	\node[state] (4) [below of=2] {4,4};
	\node[state] (23) [right of=4] {2,3};
	\node[state] (32) [right of=2] {3,2};
	\node[state] (12) [left of=1] {1,2};
	\node[state] (34) [left of=3] {3,4};
	\node[state] (14) [left of=12] {1,4};
	\node[state] (13) [left of=34] {1,3};
	\node[state] (21) [left=3mm of 14] {2,1};
	\node[state] (41) [left of=21] {4,1};
	\node[state] (43) [left=3mm of 13] {4,3};
	\node[state] (24) [right=3mm of 32] {2,4};
	\node[state] (42) [right=3mm of 23] {4,2};
	\node[state] (31) [left of=43] {3,1};
	\path[->]
	(1) edge [loop left] node [above] {$x$} (1)
	(1) edge [loop right] node [above] {$y$} (1)
	(1) edge node [left] {$a$} (3)
	(2) edge node [right] {$a$} (4)
	(2) edge [bend left=50] node [right] {$x$} (4)
	(3) edge node [below] {$x$} (4)
	(4) edge [bend left=50] node [right] {$y$} (2)
	(23) edge node [below] {$x$} (4)
    (32) edge [bend left] node [above right=0.5mm and 2mm] {$x$} (4)
    (12) edge node [right] {$a$} (34)
    (12) edge node [above] {$x$} (14)
    (21) edge node [right] {$a$} (43)
    (21) edge node [above] {$x$} (41)
    (14) edge [bend right] node [below] {$y$} (12)
    (41) edge [bend right] node [below] {$y$} (21) 
    (13) edge node [left] {$x$} (14)
    (31) edge node [left] {$x$} (41)
	;
	\end{tikzpicture}
    \caption{Synchonous composition of the suspension automaton from \figureref{fig:spec}.}
 \label{fig:comptransitive}
\end{figure}

\section{Test Cases}

In this section, we introduce a simple notion of \emph{test cases}.
The goal of these test cases is \emph{state identification}, i.e.,
to explore whether a state of the SUT, that is
reached after some initial interactions, has the same traces as the state where it
should be according to a given suspension automaton.
Our test cases are adaptive in the sense that inputs that are sent to the SUT
may depend on previous outputs generated by the SUT.
They are similar to the adaptive distinguishing sequences of Lee \& Yannakakis
\cite{leeyannakakis}, except that inputs and outputs do not necessarily alternate,
and the graph structure is a DAG rather than a tree.

\begin{definition}
A \emph{test case} is an acyclic automaton $A=(Q,T,q_0)$ such that each state $q\in Q$
enables either a single input action, or zero or more output actions.
We refer to states that enable a single input as \emph{input states}, and
states that enable at least one output as \emph{output states}.
Thus each state from a test case is either an input state, an output state, or a leaf.
\end{definition}

To each test case $A$ we associate a set of \emph{observations}: maximal traces that we may observe during a run of $A$.

\begin{definition}
For each test case $A$, $\Obs{A}$ is the set of traces that reach a leaf of $A$: $\Obs{A} = \{\sigma \in \traces{A} \mid \after{A}{\sigma} \subseteq \leaves{A}\}$.
\end{definition}

Given a suspension automaton $S$, we only want to consider test cases $A$ that are consistent with $S$ in the sense that each input  that is provided by $A$ is also specified by $S$, and conversely each output that is allowed by $S$ also occurs in $A$.

\begin{definition} \label{def:testcase}
Let $A=(Q,T,q_0)$ be a test case and $S=(Q',T',q'_0)$ a suspension automaton.
We say that $A$ is \emph{a test case for} $S$ if, for each state $(q,q')$ of $A \| S$
reachable from the initial state $(q_0, q'_0)$:
\begin{itemize}
	\item 
	if $q$ is an input state then $\inp{q} \subseteq \inp{q'}$,
	\item 
	if $q$ is an output state then $\out{q'} \subseteq \out{q}$.
\end{itemize}
We say $A$ \emph{is a test case for} state $q'\in Q'$ if $A$ is a test case for $S / q'$.
Furthermore, $A$ \emph{is a test case for} a set of states $P \subseteq Q'$ if $A$ is a test case for all $q'\in P$.
\end{definition}

\iflong
\begin{lemma}
\label{la: test cases preserved by transitions}
Suppose $A = (Q, T, q_0)$ is a test case for a set $P$ of states of suspension automaton $S$.
Suppose that $T(q_0, \mu) = q_1$, for some label $\mu$ and state $q_1$.
Then $A / q_1$ is a test case for $\after{P}{\mu}$.
\end{lemma}
\fi

If $A$ is a test case for a suspension automaton $S$ then the composition $A \| S$ is also a test case.
We can view $A \| S$ as the subautomaton of $A$ in which all outputs that are not enabled in $S$ have been pruned away.
A test case distinguishes two states, if the states enable different observable traces of the test case.

\begin{lemma}
If $A$ is a test case for a suspension automaton $S$, then the composition $A \| S$ is
also a test case for $S$, satisfying $\Obs{A \| S} \subseteq \Obs{A}$.
\end{lemma}

\begin{definition} \label{def:distinguishtestcase}
Let $A$ be a test case for states $q$ and $q'$ of suspension automaton $S$.
Then $A$ \emph{distinguishes} $q$ and $q'$ if $\Obs{A \| (S / q)} \cap \Obs{A \| (S / q')}  =  \emptyset$.
\end{definition}

\begin{example}
 The associated automaton of the CCS expression $a.(x.\Nil+y.\Nil)$ (see \exampleref{exmp:CCS}) is a test case for states 1 and 2 of the suspension automaton from \figureref{fig:spec}. Its observable traces are $\{ax,ay\}$, and it distinguishes states 1 and 2.
\end{example}

\begin{lemma} \label{lem:disttreenotcompatible}
Let $S = (Q, T, q_0)$ be a suspension automaton with $q, q' \in Q$. Then $q \not\compatible q'$ iff there exists a test case that distinguishes $q$ and $q'$.
\end{lemma}

\begin{proof}
By Lemma~\ref{compatibility reduces to validity},
$q \not\compatible q'$ iff the pair $(q,q')$ is an invalid state of $S \| S$.
By definition, this means that in the game for $S \| S$ the tester has a winning strategy $\mathit{move}$.
This strategy can be effectively computed by Algorithm~\ref{blockingalg}.
Using strategy $\mathit{move}$, we compute a test case $A$ as follows:
\begin{itemize}
	\item 
	The set of states consists of the set $P$ of invalid states of $S \| S$, extended with a
	single leaf state $l$.
	\item 
	The initial state is $(q, q')$.
	\item 
	The transition relation of $A$ is obtained by 
	(a) restricting the transition relation of $S \| S$ to $P$,
	(b) removing all input transitions, except the outgoing transitions with label
	$\mathit{move} (r,r')$ from states with $\mathit{move} (r,r') \in I$,
	(c) adding an output transition $((r,r'), x, l)$ for each $(r,r')\in P$ and
	$x \in O$ such that $\mathit{move}(r,r')=\theta$ and $(r,r')$ does not have an outgoing $x$-transition.
\end{itemize}
It is routine to check that $A$ is a test case for states $q$ and $q'$ of $S$.
We claim that $A$ distinguishes $q$ and $q'$, that is,
$\Obs{A \| S / q} \cap \Obs{A \| S / q'}  =  \emptyset$.
Because suppose $\sigma \in \Obs{A \| S / q}$.
Then $\sigma$ corresponds to a run from initial state $(q,q')$ of $A$ to leaf node $l$. By construction of $A$, $\sigma$ must be of the form $\rho x$, where
$\rho$ corresponds to a run in $A$ from $(q,q')$ to some state $(r,r')$ and $x \in \outm{r} \setminus \outm{r'}$.
This means that $A \| S / q'$ has a run with actions $\rho$ from initial state
$((q,q'),q')$ to state $((r,r'),r')$.
However, since $x \not\in \outm{r'}$,
$\sigma \not\in \Obs{A \| S / q'}$.
By a symmetric argument, we may conclude that $\sigma \in \Obs{A \| S / q'}$
implies $\sigma \not\in \Obs{A \| S / q}$.
Thus $\Obs{A \| S / q} \cap \Obs{A \| S / q'}  =  \emptyset$, as required.
\qed
\end{proof}

The following definition generalizes the notion of adaptive distinguishing sequence for
FSM's \cite{Gill62,leeyannakakis} to the setting of suspension automata.

\begin{definition}
Let $S= (Q, T, q_0)$ be a suspension automaton, $P \subseteq Q$, and $A$ a test case for $P$.
We say that $A$ is an \emph{adaptive distinguishing graph} for $P$ if,
for all $q, q' \in P$ with $q \not\compatible q'$, $A$ distinguishes $q$ and $q'$.
Test case $A$ is an \emph{adaptive distinguishing graph} for $S$ if it is an adaptive distinguishing graph for the set $Q$ of states of $S$.
\end{definition}

Just like there are FSMs without an adaptive distinguishing sequence,
there are suspension automata for which no adaptive distinguishing graph exists.
This is the case for the suspension automaton from \figureref{fig:noadg}. 
We cannot construct an adaptive distinguishing graph by choosing the root node to be an output state, since states 1 and 3 cannot be distinguished, as they both go to state 2 with their single output transition $y$.
The root also cannot be an input state for either of all inputs $a$ or $b$.
After $a$, states 1 and 2 both reach state 1, and after $b$, states 2 and 3 both reach state 3.

In the remainder of this paper, we present algorithms for constructing an adaptive distinguishing graph for $S$ from a suspension automaton $S$, if it exists. 

\begin{figure}
\centering
 \begin{tikzpicture}[shorten >=1pt,node distance=20mm,>=stealth']
	\tikzstyle{every state}=[draw=black,text=black,inner sep=1pt,minimum
	size=10pt,initial text=]
	\node[state,initial,initial where=above] (1) {1};
	\node[state] (2) [right of=1] {2};
	\node[state] (3) [right of=2] {3};
	\path[->]
	(1) edge [loop left] node [above] {$a$} (1)
	(1) edge node [above] {$y$} (2)
	(1) edge [bend left=60] node [above] {$b$} (2)
	(2) edge [bend left] node [below] {$a$} (1)
	(2) edge [bend right] node [below] {$b$} (3)
	(2) edge [loop below] node [below] {$x$} (2)
	(3) edge node [above] {$y$} (2)
	(3) edge [bend right=60] node [above] {$a$} (2)
    (3) edge [loop right] node [right] {$b$} (3)
	;
	\end{tikzpicture}
	\vspace{-5mm}
    \caption{A suspension automaton without adaptive distinguishing graph.}
 \label{fig:noadg}
\end{figure}
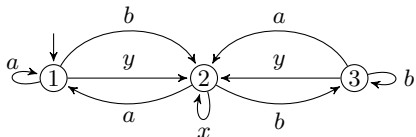

\section{Splitting Graphs}
\label{sec:algorithms}

In this section, we present the concept of a splitting graph, as well as an algorithm for constructing such a graph.
Our algorithm generalizes the algorithm of Lee \& Yannakakis \cite{leeyannakakis} for computing a splitting tree for an FSM.
In the next section, we will construct an adaptive distinguishing graph by extracting its parts from the splitting graph.
An adaptive distinguishing graph that distinguishes all incompatible state pairs, is only guaranteed to be found, if some additional requirements on the splitting graph construction are satisfied.
We will delay the discussion of adaptive distinguishing graphs to the next section, and focus on splitting graphs first.

We will first give the definition of a splitting graph, and the outer loop of our algorithm for constructing it. 
Then we define when a leaf node of a splitting graph is splittable (i.e., when child nodes can be added), and show that a splittable leaf
exists whenever some leaf contains incompatible states.
After that, we explain how to construct the child nodes for splittable leaves.

\subsection{Splitting Graph Definition}

A splitting graph for suspension automaton $S = (Q,T,q_0)$ is a directed graph in which the vertices are subsets of states of $S$;
there is a single root $Q$
, and an internal node is the union of its children.
We require that, for each edge $(v,c)$ of the splitting graph, $c$ is a proper subset of $v$; this implies that a splitting graph is a DAG.
We associate a test case $W(v)$ to each internal node $v$ and require a tight link between the observations of $W(v)$ and the children of $v$:
each observation $\sigma$ has one child $c$ that contains all states enabling $\sigma$.
As we have $|c| < |v|$, this means that, after following any trace $\sigma$ from test case $W(v)$, the states $v\setminus c$ have been distinguished from states from the states $c$.

 \begin{definition} \label{def:splitgraph}
 A \emph{splitting graph} for suspension automaton $S = (Q,T,q_0)$ is a triple $Y = (V,E,W)$ with 
 \begin{itemize}
 \item 
 $Q \in V \subseteq \mathcal{P}(Q) \setminus \emptyset$
 \item 
 $E \subseteq V \times V$ such that 
 \begin{enumerate}
 	\item
 	 $Q$ is the only root of $Y$,
 	\item 
	 $(v,w) \in E \implies v \supset w$, and
 	\item $v \in \nonleaves{Y} \implies v = \bigcup \post(v)$.
\end{enumerate} 
 \item 
$W : \nonleaves{Y} \rightarrow \Ev$ is a \emph{witness} function such that,
for all internal vertices $v$, $A_{W(v)}$ is a test case such that:\\
$\forall \sigma \in \Obs{A_{W(v)}}, \exists c \in \post(v): \enabled{c}{\sigma} = \enabled{v}{\sigma}$.
\end{itemize}
Splitting graph $Y$ is \emph{complete} if, for each leaf $l$, 
the states contained in $l$ are pairwise compatible, i.e., $\Diamond(l)$.
\end{definition}

\algorithmref{alg:splitgraphconstruction} shows the main loop for constructing a splitting graph for a given suspension automaton.
The idea is to start with the trivial splitting graph with just a single node, and then repeatedly split leaf nodes, i.e., add child nodes, until all leaves only contain pairwise compatible states. 
This means that incompatible states are in different leaves when the algorithm terminates.
Since nodes in a splitting graph are finite sets of states, and children are strict subsets of their parents, \algorithmref{alg:splitgraphconstruction} terminates after a finite number of refinements.
With $\bot$, we denote the empty function.

\begin{algorithm}[!ht]
 \caption{Constructing a splitting graph.}
  \label{alg:splitgraphconstruction}
  \SetKwInOut{Input}{Input}
  \Input{A suspension automaton $S = (Q, T, q_0)$}
  \SetKwInOut{Output}{Output}
  \Output{A complete splitting graph for $S$}  
 
    \SetKwFunction{func}{splitnode}
  $Y := (\{Q\},\emptyset,\bot)$\;
      \While{$\exists l \in \leaves{Y} : \neg \Diamond(l)$}{
	$Y :=$\func($S,Y)$\;
      }
      \Return{Y}\;         
 \end{algorithm}
 
 \subsection{Splitting Conditions}

Before we elaborate on the algorithm for the method {\tt splitnode}, we first explore what conditions should hold for a leaf $l$ to be splittable.
The formal definition of these conditions is given below in \definitionref{def:splitconditions}.

If we are lucky we can find, for each output $x \in \out{l}$, a state $q \in l$ that does not enable $x$.
In this case, observing an output allows us to distinguish at least one state from some other states.
Otherwise, we may check whether, for certain enabled inputs, or all outputs, the states of $l$ have a transition to the states of an internal node, i.e., a node that has already been split, because $l$ then may be split as well when these labels occur.
In particular, the states of $l$ can be split for some label $\mu$ if the reached node is a \emph{least common ancestor} of $\after{l}{\mu}$.
An internal node $v$ is least common ancestor for a set of states $P$ if it contains $P$ but none of its children does.

\begin{definition} \label{def:lca}
 Let $Y$ be a splitting graph for suspension automaton $S$ and let $P$ be a set of states of $S$.
 An internal node $v$ of $Y$ is a \emph{least common ancestor} of $P$ if
 $P \subseteq v$ and, for all $c \in \post(v)$, $P \not\subseteq c$.
 We write $LCA(Y,P)$ for the set of least common ancestors of $P$ contained in $Y$.
\end{definition}

Note that we can compute the set of least common ancestors for any set $P$ in a time that is linear in the size of the splitting graph.

\begin{definition} \label{def:splitconditions}
Let $Y$ be a splitting graph for suspension automaton $S$.
 \begin{enumerate}
  \item \label{outputtransitioncond} A leaf $l$ of $Y$ is \emph{splittable on output} if
  \begin{align*}
  \forall x \in \outm{l}: &(\exists q \in l: x \not\in\out{q})
  \vee LCA(Y,\after{l}{x}) \neq \emptyset
  \end{align*}
  \item \label{inputtransitioncond} A leaf $l$ of $Y$ is \emph{splittable on input} if
  \begin{align*}
  &\exists a \in \inpm{l}: LCA(Y,\after{l}{a}) \neq \emptyset
  \end{align*}
 \end{enumerate}
 A leaf $l$ of $Y$ is \emph{splittable} if it is splittable on output or splittable on input.
\end{definition}

\begin{lemma}
\label{lemma splittable leaf}
	Each incomplete splitting graph has a splittable leaf.
\end{lemma}

\begin{proof}
Let $Y$ be an incomplete splitting graph for suspension automaton $S = (Q, T, q_0)$.
Since $Y$ is incomplete, there is at least one leaf that contains a pair of incompatible states.
By \lemmaref{compatibility reduces to validity}, we have that for all states $q, q'$ of $S$,
$q \not\compatible q'$ iff $(q, q')$ is an invalid state of $S \| S$.	
Using Algorithm~\ref{blockingalg}, we may therefore compute the pairs of incompatible states of $S$ and functions $\mathit{move}$ and $\mathit{level}$ on these pairs.
Let $l$ be the leaf node that contains a pair of incompatible states $q, q'$ for which the value $\mathit{level}(q,q')$ is minimal.
We claim that $l$ is a splittable leaf of $Y$. There are three cases:
\begin{enumerate}
\item 
Suppose $\mathit{level}(q,q') = 0$.
Then, by \lemmaref{lemma move and level}(1), $(q,q')$ is a blocking state of $S \| S$.
This implies that $\outm{q} \cap \outm{q'} = \emptyset$.
But this means that, for each output action $x$, either $x \not\in \outm{q}$ or $x \not\in \outm{q'}$. Therefore,
$l$ can be split on output.
\item
Suppose $\mathit{level}(q, q') > 0$ and $\mathit{move}(q, q')=a \in I$.
Then, by \lemmaref{compatibility reduces to validity} and \lemmaref{lemma move and level}(2),
both $q$ and $q'$ enable input $a$ and, writing $r= T(q,a)$ and $r' = T(q',a)$, we have $r \not\compatible r'$,
$\{ r, r' \} \subseteq \after{l}{a}$, and $\mathit{level}(r,r')) < \mathit{level}(q,q')$.
Since none of the leaves contains a pair of incompatible states with a $\mathit{level}$ value smaller than $(q,q')$,
we know that $Y$ does not have a leaf node that contains both $r$ and $r'$.
But this implies that $LCA(Y,\after{l}{a}) \neq \emptyset$, and so $l$ can be split on input.
\item
Suppose $\mathit{level}(q,q') > 0$ and $\mathit{move}(q,q') = \theta$.
Let $x \in \outm{l}$.
If there exists an $s \in l$ such that $x \not\in\outm{s}$ then we may split on output.
Otherwise, both $q$ and $q'$ enable output $x$. Write $r= T(q,x)$ and $r' = T(q',x)$.
Then $\{ r, r' \} \subseteq \after{l}{x}$ and $r \not\compatible r'$.
By \lemmaref{compatibility reduces to validity} and \lemmaref{lemma move and level}(2),
$\mathit{level}(r,r') < \mathit{level}(q,q')$.
Since none of the leaves contains a pair of incompatible states with a $\mathit{level}$ value smaller than $(q,q')$,
we know that $Y$ does not have a leaf node containing both $r$ and $r'$.
But this implies that $LCA(Y,\after{l}{x}) \neq \emptyset$, so $l$ can be split on output. \qed
\end{enumerate}
\end{proof}

\subsection{Splitting Graph Construction}

Based on the condition of \definitionref{def:splitconditions} that holds, we assign children to splittable leaf nodes, and update 
the witness function.
This is worked out in the method {\tt splitnode} of \algorithmref{splitalg}.
The algorithm may choose nondeterministically between a split on output or a split on input, if both are possible. Such a choice is denoted with the syntax for guarded commands \cite{guardedcommands}, i.e., as the guards on lines 5 and 16, and their respective statements on lines 6-15, and 17-20.

If a leaf $l$ is split on output, then children are added for each output $x \in \out{l}$. 
If $\enabled{l}{x} \neq l$, then we add $\enabled{l}{x}$ as a child, as those are the only states from which $x$ can be observed.
We also add (i.e., by using +) the term $x.\Nil$ to the witness of $l$, as observing $x$ distinguishes states in $\enabled{l}{x}$ from states in $l\setminus\enabled{l}{x}$.

If $\enabled{l}{x} = l$, observing $x$ will not distinguish any states. We then use that there is a $v \in LCA(Y,\after{l}{x})$, which means that some states of $\after{l}{x}$ are distinguished by the witness $W(v)$. Hence, by taking output $x$, followed by $W(v)$, some states of $l$ are distinguished.
Therefore, we add $x.W(v)$ to the witness of $l$, and split $l$ in the same way $v$ was split, i.e., if $d \subseteq l$ are all the states with $\after{d}{x} \subseteq c$ for some child $c \in \post(v)$, then $d$ is a child of $l$. We call such a split an \emph{induced split}.

For splitting on some input $a$, we also use an induced split to obtain the children for $l$. Since there exists some
$v \in LCA(Y,\after{l}{a})$,  at least two states of $l$ may be distinguished by the witness constructed for $v$, after taking input $a$.
To each element of the induced split, we add all the states not enabling $a$. If we would not do this, \algorithmref{splitalg} may assign the empty set as children to a splittable leaf, such that it remains a leaf. As a consequence, \lemmaref{lem:splitgraphfromalg} and also \corollaryref{cor:completesplitgraph} then do not hold.
\iflong
This will be illustrated by \exampleref{exmp:counterexampleinputsplit}.
\fi
\corollaryref{cor:completesplitgraph} shows termination of our splitting graph construction algorithm. It follows from the consecutive application of \lemmaref{lem:splitgraphfromalg}.

\begin{definition} \label{def:inducedsplit}
 Let $Y$ be a splitting graph for suspension automaton $S$. 
 Let $v$ be an internal node of $Y$, $P$ a set of states of $S$, and $\mu \in L$, such that $\after{P}{\mu} \subseteq v$. 
 Then the \emph{induced split} of $P$ with $\mu$ to $v$ is:
 \begin{eqnarray*}
  \Pi(P,\mu,v) & = & \{ (\before{c}{\mu}) \cap P  \mid c \in \post_{Y}(v)\} \setminus \emptyset.
 \end{eqnarray*}
\end{definition}

\begin{algorithm}[!ht]
 \caption{Splitting a leaf node of a splitting graph.}
  \label{splitalg}
  \SetKwInOut{Input}{Input}
  \Input{A suspension automaton $S = (Q, T, q_0)$}
  \Input{An incomplete splitting graph $Y=(V,E,W)$ for $S$}
      \SetKwInOut{Output}{Output}
  \Output{A splitting graph $Y'$ that extends $Y$ with additional leaf nodes}
  \SetKwFunction{func}{splitnode}
  \SetKwProg{myalg}{Function}{:}{}
  \SetKwProg{FirstGuard}{if}{ $\rightarrow$}{}
  \SetKwProg{LastGuard}{[]}{ $\rightarrow$}{fi} 
  \myalg{\func$(S,Y)$}{
  $l :=$ a splittable leaf of $Y$\; 
  $C := \emptyset$\;
  $F := \Nil$\;
  \FirstGuard{$l$ splittable on output}
    {
        \ForEach{$x \in \outm{l}$}{
        \eIf{$\exists q \in l: x \not\in\out{q}$}
        {
            $C := C \cup \{ \enabled{l}{x} \}$\;
            $F := F + x.\Nil$\;
        }
        {
            Let $v \in LCA(Y,\after{l}{x})$\;
            $C := C \cup \Pi(l,x,v)$\;
            $F := F + x.W(v)$\;
        }
        }
      }
      \LastGuard{$l$ splittable on input}{ 
        Let $a \in \inp{l}$ with $LCA(Y,\after{l}{a}) \neq \emptyset$\;
      	Let $v \in LCA(Y,\after{l}{a})$\; 
      	$C := \{d \cup (l \setminus \enabled{l}{a}) \mid d \in \Pi(l,a,v) \}$\;
      	$F := a.W(v)$\;
      }
    \Return{$(V \cup C,E \cup \{(l,c) \mid c \in C\}, W \cup \{l \mapsto F \})$}\;
    }
 \end{algorithm}
 
  \begin{example} \label{exmp:splitgraphconstr}
We compute the splitting graph of the suspension automaton from \figureref{fig:spec}, using \algorithmref{alg:splitgraphconstruction}, and show the result in \figureref{fig:splitdistgraph}(left).

For the root node $\{1,2,3,4\}$, we observe that state 4 does not enable $x$, while states 2 and 3 do not enable $y$. Hence, the root is split on output, gets children $\{1,2,3\}$ and $\{1,4\}$, and witness $x.\Nil + y.\Nil$.

Node $\{1,2,3\}$ can be split on input $a$, as states 1 and 2 enable $a$, and since the root node is an LCA of $\after{\{1,2,3\}}{a}$: from $T(1,a) = 3$ and $T(2,a) = 4$, we obtain that the root node is an LCA, since $\{3,4\} \subseteq \{1,2,3,4\}$, but  $\{3, 4\} \not\subseteq \{1,2,3\}$, and $\{3, 4\} \not\subseteq \{1,4\}$.
The induced split is $\{\{1\},\{2\}\}$. We then need to add state 3 to both sets, because state 3 does not enable $a$, so node \{1,2,3\} gets children \{1,3\} and \{2,3\}.
Prepending $a$ to the witness of the root node gives us witness $a.(x.\Nil + y.\Nil)$ for $\{1,2,3\}$.

Node $\{1,4\}$ can be split on output. As state 4 does not enable $x$, we only need to find an LCA for $\after{\{1,4\}}{y} = \{1,2\}$, which is the previously split node $\{1,2,3\}$. For $x$ we have witness $x.\Nil$, and for $y$ we use the witness of $\{1,2,3\}$, so the witness for \{1,4\} is 
$x.\Nil + y.a.(x.\Nil + y.\Nil)$. 

Next, node \{1,3\} can be split on output using \{1,4\} as LCA for $x$.
Node \{2,3\} does not need to be split, as we have $2 \compatible 3$.
All other leaves are singletons, so we have obtained a complete splitting graph.
\end{example}
 
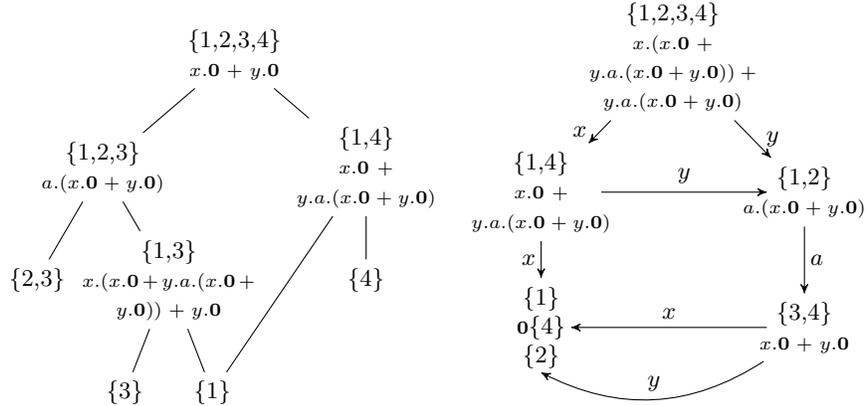
\begin{figure}[ht!]
\quad

\hspace{-7mm}
	\begin{tikzpicture}[level/.style={sibling distance=35mm/#1, level distance=15mm}]
	\tikzstyle{every node}=[align=center,text width=2.3cm]
	\node [align=center] {\{1,2,3,4\}\\\smallerfont $x.\Nil + y.\Nil$}
	child {node [align=center] {\{1,2,3\}\\\smallerfont $a.(x.\Nil + y.\Nil)$}
		child {node {\{2,3\}}
			edge from parent node {}
		}
		child {node {\{1,3\}\\\smallerfont $x.(x.\Nil + y.a.(x.\Nil + y.\Nil)) + y.\Nil$}
            child {node {\{3\}}
                edge from parent node {}
            }
            child {node (1) {\{1\}}
                edge from parent node {}
            }
			edge from parent node {}
		}
		edge from parent node {}
	}
	child {node (14) [align=center] {\{1,4\}\\\smallerfont $x.\Nil + y.a.(x.\Nil + y.\Nil)$}
		child {node {\{4\}}
			edge from parent node {}
		}
	};
	\path[-]
	(14) edge (1)
	;
	\end{tikzpicture}
	\hspace{-5mm}
	\begin{tikzpicture}[level/.style={sibling distance=35mm/#1, level distance=18mm},>=stealth']
	\tikzstyle{every node}=[align=center,text width=2.3cm]
	\node [align=center] {\{1,2,3,4\}\\\smallerfont $x.(x.\Nil + y.a.(x.\Nil + y.\Nil)) + y.a.(x.\Nil+y.\Nil)$}
	child {node (14) [align=center] {\{1,4\}\\\smallerfont $x.\Nil + y.a.(x.\Nil + y.\Nil)$}
		child {node (1) {\{1\}\\{\smallerfont$\Nil$}\{4\}\\\{2\}}
			edge from parent [->] node [left=-11mm] {$x$}
		}
		edge from parent [->] node [left=-10mm] {$x$}
	}
	child {node (12) [align=center] {\{1,2\}\\\smallerfont $a.(x.\Nil+y.\Nil)$}
		child {node (34) {\{3,4\}\\\smallerfont $x.\Nil+y.\Nil$}
			edge from parent [->] node [right=-11mm] {$a$}
		}
		edge from parent [->] node [right=-10mm] {$y$}
	};
	\path[->]
	($(14)+(0.8,0)$) edge node [above] {$y$} ($(12)+(-0.5,0)$)
	($(34)+(-0.5,0)$) edge node [above] {$x$} ($(1)+(0.4,0)$)
	(34) edge [bend left] node [above] {$y$} ($(1)+(0,-0.6)$)
	;
	\end{tikzpicture}
	\caption{Splitting graph for the suspension automaton of \figureref{fig:spec} (left), where we shorten all CCS expresions $\Nil + F$ to $F$, and 
	an adaptive distinguishing graph for the suspension automaton of \figureref{fig:spec} (right), annotated with current state sets $P$, as used in \algorithmref{alg:distgraph}.}
	\label{fig:splitdistgraph}
\end{figure}

\iflong
 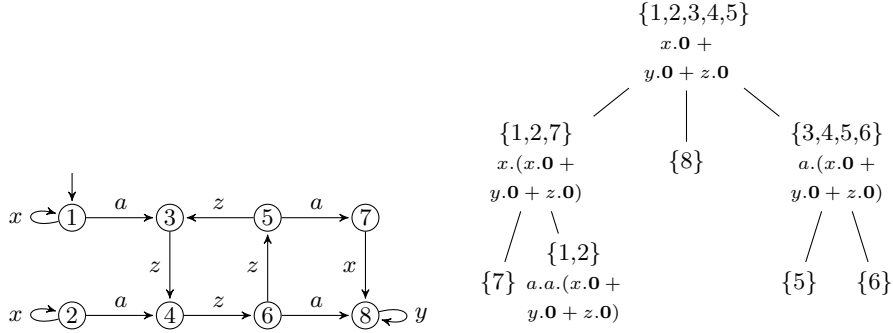
\begin{figure}
	\begin{tikzpicture}[shorten >=1pt,node distance=1.3cm,>=stealth']
	\tikzstyle{every state}=[draw=black,text=black,inner sep=1pt,minimum
	size=10pt,initial text=]
	\node[state,initial,initial where=above] (1) {1};
    \node[state] (2) [below of=1] {2};
	\node[state] (3) [right of=1] {3};
	\node[state] (4) [right of=2] {4};
	\node[state] (5) [right of=3] {5};
	\node[state] (6) [right of=4] {6};
	\node[state] (7) [right of=5] {7};
	\node[state] (8) [right of=6] {8};
	\path[->]
	(1) edge [loop left] node {$x$} (1)
	(1) edge node [above] {$a$} (3)
	(2) edge [loop left] node {$x$} (2)
	(2) edge node [above] {$a$} (4)
	(3) edge node [left] {$z$} (4)
	(4) edge node [above] {$z$} (6)
	(5) edge node [above] {$z$} (3)
	(5) edge node [above] {$a$} (7)
	(6) edge node [left] {$z$} (5)
	(6) edge node [above] {$a$} (8)
	(7) edge node [left] {$x$} (8)
	(8) edge [loop right] node {$y$} (8)
	;
	\end{tikzpicture}
	\begin{tikzpicture}[level/.style={sibling distance=20mm/#1, level distance=16mm}]
	\tikzstyle{every node}=[align=center,text width=1.3cm]
	\node [align=center] {\{1,2,3,4,5\}\\\smallerfont $x.\Nil + y.\Nil+ z.\Nil$}
	child {node [align=center] {\{1,2,7\}\\\smallerfont $x.(x.\Nil + y.\Nil+z.\Nil)$}
		child {node {\{7\}}
			edge from parent node {}
		}
		child {node {\{1,2\}\\\smallerfont$a.a.(x.\Nil + y.\Nil+z.\Nil)$}
		}
		edge from parent node {}
	}
	child {node {\{8\}}}
	child {node [align=center] {\{3,4,5,6\}\\\smallerfont $a.(x.\Nil + y.\Nil+z.\Nil)$}
		child {node {\{5\}}
			edge from parent node {}
		}
		child {node {\{6\}}
			edge from parent node {}
		}
	};
	\end{tikzpicture}
	\caption{A suspension automaton (left), and its incomplete splitting graph (right), when replacing line 19 of \algorithmref{splitalg} by $C := \Pi(l,a,v)$;}
	\label{fig:counterexampleinputsplit}
\end{figure}

\begin{example} \label{exmp:counterexampleinputsplit}
 \figureref{fig:counterexampleinputsplit} shows that using only the induced split as children, for splitting a leaf on input, results in an incomplete splitting graph. The construction of the splitting graph goes as follows. The root node \{1,2,3,4,5\} can be split on output, as each state only enables one of the three outputs $x$, $y$, and $z$: we obtain children \{\{1,2,7\},\{8\},\{3,4,5,6\}\}. Leaf \{1,2,7\} can be split on output, as $\after{\{1,2,7\}}{x} = \{1,2,8\}$ shows that we can use the root node as LCA.
 Leaf \{3,4,5,6\} cannot be split on output as $\after{\{3,4,5,6\}}{z} = \{3,4,5,6\}$, so there exists no LCA for $\after{\{3,4,5,6\}}{z}$.
 It can be split on input $a$: $\after{\{3,4,5,6\}}{a} = \{7,8\}$, so we can use the root node as LCA. Then $\Pi(\{3,4,5,6\},a,\{1,2,3,4,5\}) = \{\{5\},\{6\}\}$, so these are added as children.
 It remains to split \{1,2\}, as they are incompatible: a test case with observations $\{azzax,azzay\}$ distinguishes 1 and 2.
 Leaf \{1,2\} cannot be split on input, as $\after{\{1,2\}}{x} = \{1,2\}$, so no LCA exists. For input $a$ we find that $\after{\{1,2\}}{a}$ = \{3,4\}, and \{3,4,5,6\} is an LCA.
 However, $\Pi(\{1,2\},a,\{3,4,5,6\}) = \emptyset$, as both 3 and 4 are not contained in any child of \{3,4,5,6\}.
 Hence, we obtain $\post(\{1,2\}) = \emptyset$, which means by definition that \{1,2\} is a leaf.
 \algorithmref{alg:splitgraphconstruction} will keep trying to split \{1,2\} indefinitely, and will hence not terminate.
\end{example}
\fi

\begin{lemma} \label{lem:splitgraphfromalg}
	Algorithm~\ref{splitalg} returns a splitting graph $Y'$ for $S$, when given some splitting graph $Y$, such that one leaf $l$ of $Y$, has become an internal node in $Y'$.
\end{lemma}

\begin{proof}
	The input of Algorithm~\ref{splitalg} is a splitting graph $Y$ for $S$.
	All the algorithm does is to take a single leaf node $l$, add
	children $C$ to it, and extend the evaluation function $W$ for some witness $A$ to $l$.
	This means that it in order to prove that Algorithm~\ref{splitalg} returns a splitting graph, it suffices
	to show that 
	(a) for all $d \in C$, $\emptyset \subset d \subset l$, 
	(b) $l = \bigcup C$,
	(c) $A$ is a test case, and
	(d) $\forall \sigma \in \Obs{A}, \exists c \in C: \enabled{c}{\sigma} = \enabled{l}{\sigma}$, and
	(e) $C \neq \emptyset$.
	
	To prove (a) we inspect the three places in the algorithm where a new element $d$ was added to the set $C$ of children of $l$: line 8, line 12 and line 18:
	\begin{itemize}
		\item 
		Line 8: In this case $x \in\outm{l}$ and there exists a $q \in l$ such that output $x$ is not enabled from state $q$. This implies $\emptyset \subset d =  \enabled{l}{x} \subset l$, as required.
		\item 
		Line 12: In this case, let $d \in \Pi(l, x, v)$ for some
		$v \in LCA(Y,\after{l}{x})$.
		By definition of $\Pi$, $\emptyset \subset d$ and there is a
		$c \in \post_{Y}(v)$ such that $d = (\before{c}{x}) \cap l$.
		Note that this implies $d \subseteq l$.
		By definition of LCA, there exists a $q \in \after{l}{x}$ with
		$q \not\in c$.
		Because $q \in \after{l}{x}$, there exists a state $r \in l$ such that
		$T(r,x)=q$. Since $q \not\in c$, we know that $r \not\in \before{c}{x}$.
		Hence $\emptyset \subset d \subset l$, as required.
		\item 
		Line 18: In this case, $d = e \cup (l \setminus \enabled{l}{a})$, where
		$e \in \Pi(l, a, v)$ and $v \in LCA(Y,\after{l}{a})$.
		By definition of $\Pi$, $\emptyset \neq e$ and there is a
		$c \in \post_{Y}(v)$ such that $e = (\before{c}{a}) \cap l$.
		This implies $\emptyset \subset d \subseteq l$.
		By definition of LCA, there exist $q \in \after{l}{a}$ such that
		$q \not\in c$.
		Because $q \in \after{l}{a}$, there exists a state $r \in l$ such that
		$T(r,a)=q$. Since $q \not\in c$, we know that $r \not\in \before{c}{a}$.
		This means $r \not\in e$ and thus $r \not\in d$.
		Hence $\emptyset \subset d \subset l$, as required.
	\end{itemize}

	For proving (b), it remains to show that $l \subseteq  \bigcup C$.  Choose $q \in l$. We consider
	two cases:
	\begin{itemize}
		\item 
		A split on output was performed (line 5-15). Since $S$ is a suspension automaton,
		there is at least one output $x$ that is enabled in $q$. If there is another
		state in $l$ that does not enable $x$ then $\enabled{l}{x}$ is added to $C$
		and thus $q \in \bigcup C$, as required.
		Otherwise, sets $(\before{c}{x}) \cap l$ are added to $C$, for $c \in \post_{Y}(v)$ and some $v \in LCA(Y,\after{l}{x})$.
		Let $r = T(q,x)$. Since $\after{l}{x} \subset v$ and $v = \bigcup \post_{Y}(v)$,
		there is some $c \in \post_{Y}(v)$ with $r \in c$.
		This implies $q \in (\before{c}{x}) \cap l$ and therefore
		$q \in \bigcup C$, as required.
		\item 
		A split on input was performed (lines 16-20). In this case, the sets
		$e \cup (l \setminus \enabled{l}{a})$ are added to $C$, for $e \in \Pi(l,a,v)$,
		some input $a$ and $v \in LCA(Y,\after{l}{a})$.
		If state $q$ does not enable input $a$ then state $q$ is in each set that is
		added to $C$, and thus $q \in \bigcup C$, as required.
		Now suppose $q$ enables input $a$. Let $r = T(q,a)$.
		Then $r \in \after{l}{a}$ and thus $r \in v$. Since $v = \bigcup \post_{Y}(v)$,
		there is some $c \in \post_{Y}(v)$ with $r \in c$.
		Therefore, $q \in \before{c}{x}) \cap l \in \Pi(l,a,v)$, and therefore
		$q \in \bigcup C$, as required.	
	\end{itemize}

	For proving (c) we again consider the two cases of splitting on output or input:
	\begin{itemize}
	 \item If a split on output was performed, then root of $A$ is an output state, as each observation has an output prefix: on line 9 or 13 either $x.\Nil$ or $x.W(v)$ for some $v \in LCA(Y,\after{l}{x})$ are added to $A$. Since $\Nil$ is a test case, and $A_{W(v)}$ is a test case since $v$ is an internal node of $Y$, $A$ is also a test case.
	 \item If a split on input was performed, then the root of $A$ is an input state, as it enables a single input according to line 20: $A = A_{a.W(v)}$ for some $v \in LCA(Y,\after{l}{x})$. As $A_{W(v)}$ is a test case since $v$ is an internal node of $Y$, $A$ is also a test case. 
	\end{itemize}
	
	For proving (d), we inspect the three places in the algorithm where children were added to $C$, and where witness observations were added to $A$.
	We will show that for each added observation $\sigma$, a child $d$ constructed at the same place can be used to prove $\enabled{d}{\sigma} = \enabled{l}{\sigma}$.
	\begin{itemize}
	 \item On lines 8 and 9, a child $d = \enabled{l}{x}$ was added to $C$, and observation $x$ was added to $A$. 
	 Hence, for $x \in \Obs{A}$ we have child $d$ with $\enabled{d}{x} = \enabled{l}{x}$.
	 \item On lines 12 and 13, children $d \in \Pi(l,x,v)$ are added to $C$, and observations $x\sigma$ are added to $A$ for all $\sigma \in \Obs{A_{W(v)}}$, using some $v \in LCA(Y,\after{l}{x})$.
	 Since $v$ is an internal node of $Y$, there is a $c \in \post(v)$ such that $\enabled{c}{\sigma} = \enabled{v}{\sigma}$.
	 If $\before{c}{x} \cap l = \emptyset$, then it holds that $(\after{l}{x}) \cap c = \emptyset$, so from $\after{l}{x} \subseteq v$ (by $v \in LCA(Y,\after{l}{x}$) it then follows that $\enabled{\after{l}{x}}{\sigma} = \enabled{l}{x\sigma} = \emptyset$. Hence any $d \in \Pi(l.x.v)$ can be used to show $\enabled{d}{x\sigma} = \enabled{l}{x\sigma}$ as $d \subseteq l$.
	 Else, there is some $d \in \Pi(l,x,v)$ with $d = (\before{c}{x}) \cap l$.
	 Let $e = c\setminus (\after{d}{x})$, and observe that $e \cap (\after{l}{x}) = \emptyset$.
	 From $\enabled{c}{\sigma} = \enabled{v}{\sigma}$ and $\after{l}{x} \subseteq v$ it then follows that $\enabled{(\after{d}{x}) \cup e}{\sigma} = \enabled{(\after{l}{x}) \cup (v \setminus (\after{l}{x}))}{\sigma}$, so $\enabled{\after{d}{x}}{\sigma} = \enabled{\after{l}{x}}{\sigma}$.
	 It follows that $\enabled{d}{x\sigma} = \enabled{l}{x\sigma}$.
	 \item On lines 19 and 20 children $d\cup(l\setminus\enabled{l}{a})$ for all $d \in \Pi(l,a,v)$ are assigned to $C$, and observations $a\sigma$ are added to $A$ for all $\sigma \in \Obs{A_{W(v)}}$, using some $v \in LCA(Y,\after{l}{a})$. Again, since $v$ is an internal node of $Y$, there is a $c \in \post(v)$ such that $\enabled{c}{\sigma} = \enabled{v}{\sigma}$.
	 If $\after{l}{a} \cap c = \emptyset$, then it follows, with similar arguments as for lines 12 and 13, that $\enabled{\after{l}{a}}{\sigma} = \enabled{l}{a\sigma} = \emptyset$. 
	 Since $\enabled{l\setminus\enabled{l}{a}}{a} = \emptyset$, and hence also $\enabled{l\setminus\enabled{l}{a}}{a\sigma} = \emptyset$, we can use any child $e$ from line 12 to show $\enabled{e}{a\sigma} = \enabled{l}{a\sigma}$.
	 Else, there is some $d \in \Pi(l,a,v)$ with $d = (\before{c}{a}) \cap l$. With the same reasoning as for lines 12 and 13, we obtain $\enabled{d}{a\sigma} = \enabled{l}{a\sigma}$.
	 By again using that $\enabled{l\setminus\enabled{l}{a}}{a\sigma} = \emptyset$, we obtain $\enabled{d \cup (l\setminus\enabled{l}{a})}{a\sigma} = \enabled{l}{a\sigma}$.
	\end{itemize}
	
	For proving (e) we consider the two cases of splitting on output or input:
	\begin{itemize}
	 \item Suppose an output split is performed. The body of the for-loop on lines 7-14 is then executed at least once, since the algorithm only accepts suspension automata, so each state is non-blocking, and consequently $|\out{l}| \ge 1$.
	 Hence, suppose that the for-loop is executed for some $x \in \out{l}$.
     To prove that $C \neq \emptyset$, we now need to show that $\{\enabled{l}{x}\} \neq \emptyset$ (line 8), and that $\Pi(l,x,v) \neq \emptyset$ (line 12), using some $v \in LCA(Y,\after{l}{x}$ (line 11).
     
     For line 8, we use from (a) that $\emptyset \subset \enabled{l}{x}$, so $\{\enabled{l}{x}\} \neq \emptyset$.
     
     For line 12, we need to prove that there exists a $c \in \post(v)$ such that $(\before{c}{x}) \cap l \neq \emptyset$.
     Because there is some $v \in LCA(Y,\after{l}{x})$, we have $\after{l}{x} \subseteq v$. Since $x \in \out{l}$, there is a $q \in \after{l}{x}$, so $\emptyset \subset \before{q}{x} \subseteq l$. Because $v = \bigcup\post(v)$, there is a $c \in \post(v)$ with $q \in c$. Hence, $\before{q}{x} \subseteq \before{c}{x}$. It then follows that $(\before{c}{x}) \cap l \neq \emptyset$. 
     \item Suppose an input split is performed for some input $a$. We then have a $c \in \post(v)$ with $(\before{c}{a}) \cap l \neq \emptyset$, for the same reasons as given for line 12. Consequently, $\Pi(l,a,v) \neq \emptyset$. Adding the (possibly empty) set $l\setminus\enabled{l}{a}$ to each element of $\Pi(l,a,v)$ results in a non-empty set $C$. \qed
	\end{itemize}

\end{proof}

\begin{corollary} \label{cor:completesplitgraph}
 \algorithmref{alg:splitgraphconstruction} returns a complete splitting graph for $S$.
\end{corollary}

The algorithm of \cite{leeyannakakis} constructs a splitting tree in polynomial time, because leaves of a node form a partition of that node.
Our splitting graphs do not have this property.
Clearly, a splitting graph for a suspension automaton with $n$ states cannot have more than $2^n$ nodes, as the set of nodes is a subset of $\mathcal{P}(Q)\setminus \emptyset$ by \definitionref{def:splitgraph}.
 For $n \in \mathbb{N}$ with $n \geq 3$, consider suspension automaton $S_n = (\{1, \dots, n\},T_n ,1)$, where $T_n$ consists of the following output transitions:
 \begin{eqnarray*}
  T_n & = & \{(n,n,1)\} \cup \{(s,x,s+1) \mid s \in \{1,\dots,n-1\}, x \in \{1,\dots,n-1\}, s \neq x \}.
 \end{eqnarray*}
 \figureref{fig:expaut} depicts suspension automata $S_n$ for $n = 3,4,5$.
  We can prove \lemmaref{lem:expsplitgraph} by showing that $S_n$ has a splitting graph with $2^{n-1}$ nodes. 
 
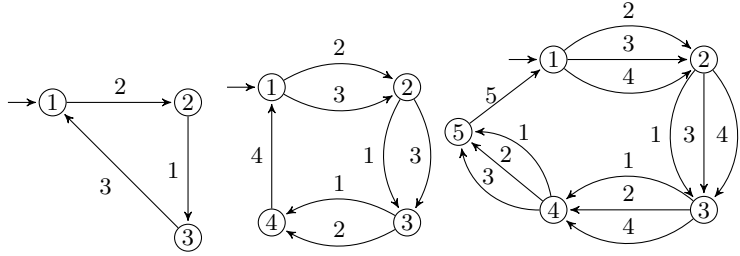
\begin{figure}[ht!]
	\centering
	\begin{tikzpicture}[shorten >=1pt,node distance=1.8cm,>=stealth']
	\tikzstyle{every state}=[draw=black,text=black,inner sep=1pt,minimum
	size=10pt,initial text=]
	\node[state,initial,initial where=left] (1) {1};
	\node[state] (2) [right of=1] {2};
	\node[state] (3) [below of=2] {3};
	\path[->]
	(1) edge node [above] {2} (2)
	(2) edge node [left] {1} (3)
	(3) edge node [below left] {3} (1)
	;
	\end{tikzpicture}
	\begin{tikzpicture}[shorten >=1pt,node distance=1.8cm,>=stealth']
	\tikzstyle{every state}=[draw=black,text=black,inner sep=1pt,minimum
	size=10pt,initial text=]
	\node[state,initial,initial where=left] (1) {1};
	\node[state] (2) [right of=1] {2};
	\node[state] (3) [below of=2] {3};
	\node[state] (4) [below of=1] {4};
	\path[->]
	(1) edge [bend left] node [above] {2} (2)
	(1) edge [bend right] node [above] {3} (2)
	(2) edge [bend right] node [left] {1} (3)
	(2) edge [bend left] node [left] {3} (3)
	(3) edge [bend right] node [above] {1} (4)
	(3) edge [bend left] node [above] {2} (4)
	(4) edge node [left] {4} (1)
	;
	\end{tikzpicture}
	\begin{tikzpicture}[shorten >=1pt,node distance=2cm,>=stealth']
	\tikzstyle{every state}=[draw=black,text=black,inner sep=1pt,minimum
	size=0pt,initial text=]
	\node[state,initial,initial where=left] (1) {1};
	\node[state] (2) [right of=1] {2};
	\node[state] (3) [below of=2] {3};
	\node[state] (4) [below of=1] {4};
	\node[state] (5) [below left=0.7cm and 1cm of 1] {5};
	\path[->]
	(1) edge [bend left=40] node [above] {2} (2)
	(1) edge node [above] {3} (2)
	(1) edge [bend right=40] node [above] {4} (2)
	(2) edge [bend right=40] node [left] {1} (3)
	(2) edge node [left] {3} (3)
	(2) edge [bend left=40] node [left] {4} (3)
	(3) edge [bend right=40] node [above] {1} (4)
	(3) edge node [above] {2} (4)
	(3) edge [bend left=40] node [above] {4} (4)
	(4) edge [bend right=40] node [above] {1} (5)
	(4) edge node [above] {2} (5)
	(4) edge [bend left=40] node [above] {3} (5)
	(5) edge node [left] {5} (1)
	;
	\end{tikzpicture}
	\caption{Suspension automaton $S_n$ for $n=3$, $n=4$, and $n=5$}
	\label{fig:expaut}
\end{figure}

\begin{lemma} \label{lem:expsplitgraph}
Let $S$ be a suspension automaton with $n$ states. Then a
splitting graph returned by \algorithmref{alg:splitgraphconstruction} has $\mathcal{O}(2^n)$ nodes. This bound is tight.
\end{lemma}
\begin{proof}
We already showed that a splitting graph has at most an exponential number of states.
We will now prove that \algorithmref{alg:splitgraphconstruction} returns a splitting graph with exactly $2^{n-1}$ nodes for suspension automaton
$S_n$ with $n \ge 3$:

We first note that different states are pairwise incompatible, since we can easily construct a test case identifying any of the states: observing output $n$, after having observed $i$ (other) outputs, means that the test case was executed from state $n-i$. Consequently, if a node of the split graph contains more than 1 state, it has children.
 
  The root node is split on output, so it has children for all size $n-2$ subsets of $\{ 1,\ldots, n-1 \}$, and it has child $\{n\}$.
 We now show that the split graph has nodes for all non-empty subsets of $\{ 1,\ldots, n-1 \}$, except trivial subset $\{ 1,\ldots, n-1 \}$.
 
Suppose we have a non-trivial subset $s$ of $\{ 1,\ldots, n-1 \}$ with at least two elements.
For all $x \in s$ state $x$ does not enable output $x$, but all other states of $s$ do, 
so we obtain child $s \setminus \{x\}$ by a split on output $x$.
By repeatedly removing a single element by splitting on that element, we can show that the split graph
contains a node for \emph{any} nonempty, non-trivial subset of $\{ 1,\ldots, n-1 \}$.
There are $2^{n-1} -2$  nonempty, non-trivial subsets of $\{ 1,\ldots, n-1 \}$.
In addition, the split graph also has nodes $\{ 1,\ldots, n \}$ and $\{ n \}$.
Hence, in total the splitting graph has $2^{n-1}$ nodes. \qed
\end{proof}

\section{Extracting Test Cases from a Splitting Graph}
\label{sec:distgraph}


\algorithmref{alg:distgraph} retrieves CCS terms, of which the associated automata are test cases that distinguish states. The algorithm ``concatenates'' several CCS terms while keeping track of the current set of states.
Each CCS term ensures that one state is distinguished from the rest because it lacks some output.
We compute the current states for the leaves of the CCS term, and attach another CCS term to this leaf, if the current set of states consists of some incompatible pair of states.
Hence in total, the automaton of the resulting CCS term distinguishes multiple pairs of states.
 
 \begin{algorithm}[!ht]
 \caption{Retrieving a test case from a splitting graph}
  \label{alg:distgraph}
  \SetKwInOut{Input}{Input}
  \Input{A suspension automaton $S = (Q, T, q_0)$}
  \Input{A complete splitting graph $Y=(V, E, W)$ for $S$}
  \SetKwFunction{func}{compDG}
  \SetKwProg{myalg}{Function}{:}{}
  \func($S,Y,Q,\Nil$)\;
  where\\
  \myalg{\func($S,Y,P,F$)}{
    \eIf{$\Diamond(P)$}{
    \Return{$F$}
    }{
        \uIf{$F = \Nil$}{
           Let $v \in LCA(Y,P)$\;
            \Return{\func($S,Y,P,W(v)$)}
        }\uElseIf{$F = \mu.F_1$ \textnormal{for some CCS term} $F_1$}{
            \Return{$\mu$.\func($S,Y,\after{P}{\mu},F_1$)}
        }\ElseIf{$F = F_1 + F_2$ \textnormal{for some CCS terms} $F_1,F_2$}{
            \Return{\func($S,Y,P,F_1$) $+$ \func($S,Y,P,F_2$)}
        }
    }
  }
  \SetKwInOut{Output}{Output}
\Output{A CCS term $F$ such that, for each $\sigma \in \Obs{A_F}$, $\Diamond(\after{Q}{\sigma})$.}
\end{algorithm}

\begin{example}
We construct the adaptive distinguishing graph for the suspension automaton from \figureref{fig:spec}, using the splitting graph from \figureref{fig:splitdistgraph}, which also shows the result of this example.
\algorithmref{alg:distgraph} starts with $P = \{1,2,3,4\}$ and $F = \Nil$. 
Hence, we search for a least common ancestor for $Q$. This will be the root node of the splitting graph, with witness $x.\Nil + y.\Nil$.

The function is then called with $F = x.\Nil + y.\Nil$, and will result in two recursive calls of the function on line 13 for $P = \{1,2,3,4\}$ and $F=x.\Nil$, and $P = \{1,2,3,4\}$ and $F=y.\Nil$ respectively.
In the first case, the condition of line 10 holds, and we the function is called for $P = \after{\{1,2,3,4\}}{x} = \{1,4\}$ and $F = \Nil$, which means that lines 8-9 are executed next, using the only LCA for \{1,4\}, namely \{1,4\}. 

The algorithm will then do some more recursive calls, checking whether the witness of \{1,4\} must be extended further to distinguish more states. This will not be the case, because only singleton sets are reached at the leaves of the witness, and $\compatible\{q\}$ holds for any state $q$, since $\compatible$ is reflexive. 
Hence, we need to prepend $x$ to the witness of \{1,4\} to obtain the  left term of the + operator of the resulting CCS term of the algorithm:  $x.(x.\Nil + y.a.(x.\Nil + y.\Nil))$.

As $P = \after{\{1,2,3,4\}}{y} = \{1,2\}$, its LCA \{1,2,3\} will be used to complete the construction of the right term of the + operator of the result.

The associated automaton of the resulting CCS term is an adaptive distinguishing graph for the suspension automaton, as it distinguishes all incompatible state pairs.
\end{example}

\begin{lemma} \label{lem:algterminates}
\algorithmref{alg:distgraph} terminates and outputs a CCS term $F$ that denotes a test case satisfying,
for each $\sigma \in \Obs{A_F}$, $\Diamond(\after{Q}{\sigma})$.
\end{lemma}

\begin{proof}
Let $S=(Q,T,q_0)$ be the suspension automaton, and $Y$ the splitting graph for $S$, that we provide to \algorithmref{alg:distgraph}.
We note that all computations are atomic, or reducing the size of the CCS expression before making a recursive call, except line 8.
However, LCAs can be computed straightforwardly: start at the root, if it is not an LCA, continue with the children containing the set of states the LCA is computed for, and repeat. This procedure always succeeds in finign an LCA, due to the following argument.
Any set of states, with at least two incompatible states, has a least common ancestor in the splitting graph, as the leaves of $Y$ are sets of mutually compatible states, its root node contains all the states from $S$, and all the states of a non-leaf are contained in at least one of its children, by \definitionref{def:splitgraph}.

By construction, \algorithmref{alg:distgraph} follows the labels of each $\sigma \in \Obs{A_{W(v)}}$ for nodes $v$ obtained on line 8.
By the property from \definitionref{def:splitgraph} that $\enabled{c}{\sigma} = \enabled{v}{\sigma}$, and $c \subset v$, we see that $|P| > |\after{P}{\sigma}|$, so after visiting line 8 at most $|Q|-1$ times, set $P$ will only contain mutually compatible states. \qed
\end{proof}

\algorithmref{alg:distgraph} does not always construct an adaptive distinguishing graph for all incompatible state pairs. To ensure this, it must be able to select an ``injective'' splitting node as LCA on line 8.
This will guarantee that a transition never maps two incompatible states to two compatible states (which cannot be distinguished any more), or that an input is used that is not enabled in some states.

\begin{definition} \label{def:injective} 
Let $S = (Q,T,q_0)$ be a suspension automaton, $P \subseteq Q$ a set of states, and $\mu \in L$ a label.
Then $\mu$ is \emph{injective} for $P$ if
\begin{eqnarray*}
\forall q,q' \in P: q \not\compatible q' & \implies &
T(q, \mu) \downarrow \wedge\ T(q',\mu) \downarrow \wedge\ T(q,\mu) \not\compatible T(q',\mu) \\
& &\vee \mu \in O \setminus (\out{q} \cap \out{q'})
 \end{eqnarray*}
\end{definition}

Analogous to the result of  \cite{leeyannakakis}, Theorem~\ref{thm:alldisting} asserts that if an adaptive distinguishing graph exists our algorithms will find it, provided there are no compatible states. This last assumption is motivated in \exampleref{exmp:noadgcomp}.
\iflong
We first need to establish the following lemma.
\begin{lemma}
\label{la: adg preserved by transitions}
Let $S$ be a suspension automaton such that all pairs of distinct states are incompatible.
Suppose $A = (Q, T, q_0)$ is an adaptive distinguishing graph for a set $P$ of states of $S$.
Suppose that $T(q_0, \mu) = q_1$, for some label $\mu$ and state $q_1$.
Then $\mu$ is injective for $P$ and $A / q_1$ is an adaptive distinguishing graph for $\after{P}{\mu}$.
\end{lemma}


\fi

\begin{theorem}
\label{thm:alldisting}
 Let $S$ be a suspension automaton such that all pairs of distinct states are incompatible. 
 Then $S$ has an adaptive distinguishing graph if and only if, 
 during construction of a splitting graph $Y$ for $S$, \algorithmref{splitalg} can and does only perform injective splits,
 that is, whenever \algorithmref{splitalg} splits a leaf $l$ on output, then
 $x$ is injective for $l$, for all $x \in \outm{l}$, and whenever it splits a leaf $l$ on input $a$, then
 $a$ is injective for $l$.   Moreover, in this case \algorithmref{alg:distgraph} constructs an adaptive distinguishing graph for $S$, when $Y$ is given as input.
\end{theorem}

\begin{proof}
Let $S = (Q, T, q_0)$.

($\impliedby$)
Suppose splitting graph $Y = (V, E, W)$ for $S$ has been constructed using injective splits only.
Then, for each internal node $v$ of $Y$, $A_{W(v)}$ is a test case for $v$: inputs performed by the test case $A_{W(v)}$ will be enabled in all the
corresponding states of $S$.
This means that also the CCS term $F$ computed from $Y$ by \algorithmref{alg:distgraph} will correspond to a test case for the set $Q$ of states of $S$.
Since all the splits in $Y$ are injective, we have that for any pair $q, q'$ of incompatible states of $S$, and for any observation $\sigma$ of $A_F$ that is enabled in both $q$ and $q'$, the unique state in $\after{q}{\sigma}$ is incompatible with the unique state in
$\after{q'}{\sigma}$. But since, by construction, $\after{Q}{\sigma}$ only contains mutually compatible states, for each observation $\sigma$ of
$A_F$, we conclude that $A_F$ distinguishes $q$ and $q'$.  Therefore, $A_F$ is an adapaptive distinguishing graph for $S$.

 ($\implies$)
 Suppose $A = (Q', T', q'_0)$ is an adaptive distinguishing graph for $S$. 
 
 Let $Y$ be an incomplete splitting graph. We show that $Y$ has a leaf for which an injective split exists.
 
 Assume w.l.o.g.\ that $A$ is a tree (any DAG can be unfolded into a tree).
 We associate to each node $r$ of $A$ a \emph{height}, which is the length of the maximal path from $r$ to a leaf.
 Also, we associate to each node of $r$ a set of states from $S$ called the \emph{current set}: the current set of $q'_0$ is $Q$, and if the current set of state $r$ is $P$ and $T'(r, \mu) = r'$ then the current set of $r'$  is $\after{r}{\mu}$.
 Lemma~\ref{la: adg preserved by transitions} implies that if the current set of $r$ equals $P$, $A / r$ is an adaptive distinguishing graph
 for $P$.
 
 Now, amongst the leaves of $Y$ that contains a maximal number of states, choose a leaf $l$ that
 is contained in the current set $P$ of a node $r$ of $A$ with minimal height. We consider two cases:
 \begin{itemize}
 \item 
 $r$ is an input state of $A$. Then $r$ enables a single input action $a$.
 Let $T'(r,a) = r'$. Then the current set of $r'$ is $\after{P}{a}$ and the height of $r'$ is less than the height of $r$.
 By Lemma~\ref{la: adg preserved by transitions}, $a$ is injective for $P$.
 By definition of injectivity, $a$ is also injective for subset $l$ of $P$.
 Since all pairs of distinct states of $S$ are incompatible, the number of states in $\after{l}{a}$ equals the number of elements of $l$.
 Moreover, since $\after{l}{a}$ is contained in $\after{P}{a}$, and
 amongst the leaves of $Y$ that contains a maximal number of states $l$ is contained in the current set of a node with minimal height,
 $\after{l}{a}$ is not contained in any leaf of $Y$.
 Thus leaf $l$ is splittable on input $a$, and this split is injective.
 \item
 $r$ is an output state of $A$. Suppose $x \in \outm{l}$.
 Then either there is a $q \in l$ such that $x \not\in\outm{q}$, or 
 the number of states in $\after{l}{x}$ equals the number of elements of $l$ and $\after{l}{x}$ is not contained in any leaf of $Y$.
 This means that $l$ is splittable on output, with a split that is injective for each output $x$. \qed
 \end{itemize}

\end{proof}

\begin{example} \label{exmp:noadgcomp}
Without the assumption that there are no compatible state pairs, Theorem~\ref{thm:alldisting} does not hold.
The suspension automaton $S$ of Figure~\ref{fig:completeness fails with compatible state pairs} has an adaptive distinguishing graph, but
our algorithm does not find it.
Note that states $2$ and $3$ are compatible, and also states $6$ and $7$ are compatible.
An adaptive distinguishing graph for $S$ is denoted by CCS term
$x . a. b. (z . \Nil + t. \Nil) + y . a. b. (z . \Nil + t. \Nil) + z. \Nil + t. \Nil$.
When we construct a splitting graph for $S$, the set of all states $\{ 1, 2, 3, 4, 5, 6, 7, 8 \}$ will be split on output,
resulting in children $\{ 1 \}$, $\{2, 3, 4 \}$, $\{ 5 \}$ and $\{ 6, 7, 8 \}$.
Now a split of $\{ 2, 3, 4 \}$ on input $b$ is not injective and a split on input $a$ is not possible since the set of LCAs is empty.
Similarly, there is no injective split of $\{ 6, 7, 8 \}$.
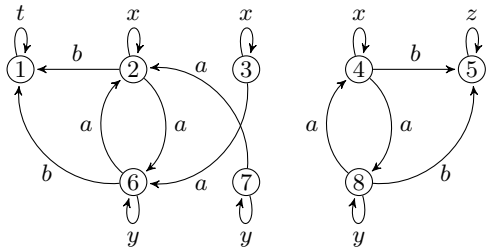
\begin{figure}
	\centering
	\begin{tikzpicture}[shorten >=1pt,node distance=1.5cm,>=stealth']
	\tikzstyle{every state}=[draw=black,text=black,inner sep=1pt,minimum
	size=10pt,initial text=]
	\node[state] (1) {1};
	\node[state] (2) [right of=1] {2};
	\node[state] (3) [right of=2] {3};
	\node[state] (4) [right of=3] {4};
	\node[state] (5) [right of=4] {5};
	\node[state] (6) [below of=2] {6};
	\node[state] (7) [below of=3] {7};
	\node[state] (8) [below of=4] {8};
	\path[->]
	(1) edge [loop above] node {$t$} (1)
	(2) edge [loop above] node {$x$} (2)
	(3) edge [loop above] node {$x$} (3)
	(4) edge [loop above] node {$x$} (4)
	(5) edge [loop above] node {$z$} (5)
	(6) edge [loop below] node {$y$} (6)
	(7) edge [loop below] node {$y$} (7)
	(8) edge [loop below] node {$y$} (8)
	(2) edge [bend left=50] node [right] {$a$} (6)
	(3) edge [bend left=50] node [below right=2mm and -4.5mm] {$a$} (6)
	(4) edge [bend left=50] node [right] {$a$} (8)
	(6) edge [bend left=50] node [left] {$a$} (2)
	(7) edge [bend right=50] node [above right=2mm and -4.5mm] {$a$} (2)
	(8) edge [bend left=50] node [left] {$a$} (4)
	(2) edge node [above] {$b$} (1)
	(6) edge [bend left=50] node [below] {$b$} (1)
	(4) edge node [above] {$b$} (5)
	(8) edge [bend right=50] node [below] {$b$} (5)
	;
	\end{tikzpicture}
	\vspace{-3mm}
	\caption{Theorem~\ref{thm:alldisting} fails in presence of compatible state pairs.}
	\label{fig:completeness fails with compatible state pairs}
	
\end{figure}
\end{example}

\section{Experimental Results on a Case Study}
\label{sec:experiments}



In \cite{learningESM}, an FSM model, with over 10.000 states,  was learned of an industrial piece of software, called the Engine Status Manager (ESM).
%
During the learning process,
testing against the ESM posed a significant challenge: it turned out to be extremely difficult to find counterexamples for hypothesis models.
Initially, existing conformance testing algorithms were used to find counterexamples
for hypothesis models (random walk, W-method, Wp-method, etc), but for larger
hypothesis models these methods were unsuccessful. 
However, adaptive distinguishing sequences as in \cite{leeyannakakis}, augmented with  additional pairwise distinguishing sequences for states not distinguished by the adaptive sequence, were able to find the required counterexamples.
Therefore, the ESM models are good candidates to show the strength of the adaptive distinguishing graphs of this paper too.

Of course, applying our adaptive distinguishing graphs directly on the Mealy machine models, would not show our capability to handle the more expressive suspension automata.
We therefore transformed the FSM models in such a way that they exhibit output nondeterminism.
We first split all Mealy $i/o$ transitions in two consecutive transitions $i$ and $o$, and added a self-loop output transition `quiescence' (denoting absense of response) to all states only having input transitions, to make it non-blocking.
To  ensure determinism, information about data parameters from the ESM was added to the labels of the Mealy machine in \cite{learningESM}. For our experiments, we removed this information again, resulting in suspension automata with states with  multiple outgoing output transitions.

For performance reasons, we reduced the Mealy machine model with a subalphabet, before applying the transformation steps described above, i.e., we removed all $i/o$ transitions with $i$ not in the subalphabet.
We obtained these subalphabets from \cite{smeenkthesis}, which contains a figure displaying interesting subalphabets based on domain knowledge.
\tableref{tab:graphstats} shows that the resulting suspension automata still have a significant size.
%

We applied the algorithms of this paper to obtain a splitting graph and an adaptive distinguishing graph.
The splitting graph was constructed as in \algorithmref{splitalg}, so without requiring injectivity of the used labels. However,
in the construction of the adaptive distinguishing graph (\algorithmref{alg:distgraph}) we chose on line 8 an LCA which was injective for the most pairs of states.

\tableref{tab:graphstats} shows that there are many pairs of incompatible states to distinguish. However, the number of nodes of the splitting graph are in the order of magnitude of the number of states of the suspension automaton, and the longest observable trace (i.e., the depth) of the adaptive distinguishing graphs is not long at all.
Moreover, over 99\% of the pairs of incompatible states are distinguished by the adaptive distinguishing graph.
This indicates that the adaptive distinguishing graphs, although constructed from a non-injective splitting graph, can be very effective in testing.

\begin{table}
 \begin{tabular}{r|c|c|c|c|c|}
  \begin{minipage}{2.5cm}\centering Subalphabet \end{minipage}
  & \begin{minipage}{1.2cm}\centering Number of states \end{minipage}                                      
  & \begin{minipage}{2cm}\centering Pairs of compatible states \end{minipage} 
  & \begin{minipage}{1.25cm}\centering Nodes in splitting graph \end{minipage} 
  & \begin{minipage}{1.6cm}\centering Depth distinguishing graph \end{minipage} 
  & \begin{minipage}{2cm}\centering Incompatible pairs not distinguished \end{minipage} \\\hline
  {\smallerfont InitIdleSleep} & 1616 & 16638 (0.64\%) & 1121 & 33 & 1145 (0.044\%) \\\hline
  {\smallerfont InitIdleStandbyRunning} & 2855 & 14171 (0.17\%) & 2082 & 33 & 2183 (0.027\%)\\\hline
  {\smallerfont InitIdleStandbySleep} & 3168 & 25974 (0.26\%) & 2226 & 33 & 3826 (0.038\%)\\\hline
  {\smallerfont InitIdleStandbyLowPower} & 2614 & 13834 (0.20\%) & 1809 & 33 & 2920 (0.043\%)\\\hline
   {\smallerfont  InitError} & 2649 & 373427 (5.3\%) & 3097 & 35 & 17972 (0.27\%) \\\hline
 \end{tabular}\vspace{2mm}
 \caption{Computation statistics}
 \label{tab:graphstats}
\end{table}

\iflong
To further explore the structure of the adaptive distinguishing graph, we computed the \emph{size} of each leaf: the number of automaton states, that enable the observable trace to that leaf.
We note that this includes states compatible to some of the automaton states. Additionally, states may enable multiple observable traces, and hence a single state may increase the size of several leaves.
\figureref{fig:colgraphs} shows the results: the x-axis displays all leaf sizes, and a column of some subalphabet shows the number of leaves of this size (y-axis).
We see that the majority of leaves are of small size, while leaves of larger size occur less. We see that subalphabet InitError has the most large leaves, which could explain the adaptive distinguishing graph's relatively large number of pairs of incompatible states not distinguished.

\begin{figure}[ht!]
\includegraphics[scale=0.63]{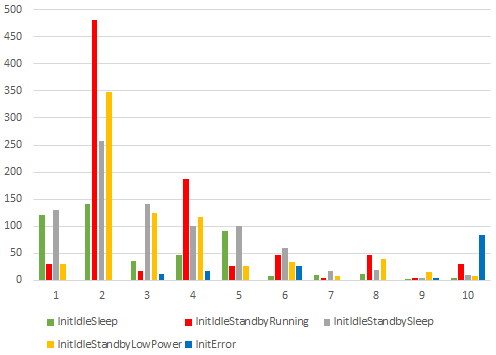}\\
\includegraphics[scale=0.63]{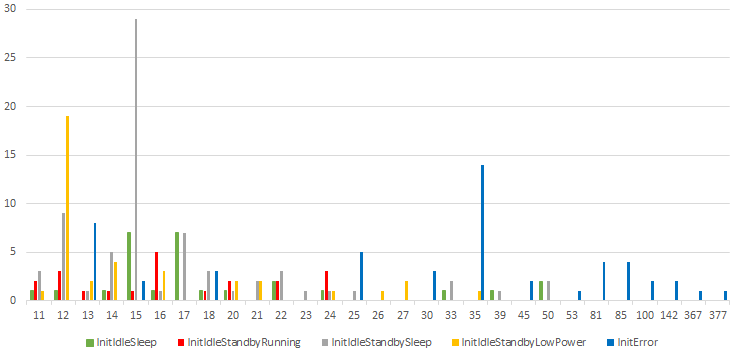}
\caption{Leaf sizes}
\label{fig:colgraphs}
\end{figure}
\fi

\section{Conclusions and Future Work}

We studied the state identification problem for suspension automata, generalizing results from \cite{leeyannakakis}. We presented algorithms to construct test cases that distinguish all incompatible state pairs, if possible, or many, if not. Experiments suggest that this approach is quite effective.

We see several directions for future research.
First, though we did apply our algorithms to instances of an industrial benchmark, we would like to apply it to different case studies as well, to further explore the applicability of our approach. We note however that there are not that many (large) LTS benchmarks available.

An open problem is to give a bound on the depth of the distinguishing graph that our algorithms constructs. 
For FSMs, a quadratic bound is known \cite{leeyannakakis}, with examples to show it is tight \cite{Sokolovskii1971,leeyannakakis}.
These examples extend to our setting, as we generalize from the FSM setting, but the proof for the quadratic bound on adaptive distinguishing sequences from \cite{leeyannakakis} does not.


If our algorithm returns an adaptive distinguishing graph that does not distinguish all incompatible state pairs, the question remains how to efficiently distinguish these remaining states.
Graphs distinguishing pairs of states can be obtained directly from our splitting graph, or by computing them as in \cite{ncompletejournal}, but distinguishing all remaining pairs results in a large overhead compared to the small size of the distinguishing graph we obtained in our experiments.
On the one hand, we can optimize the obtained distinguishing graph by improving the splitting graph's quality by applying heuristics that optimize the choice of labels for splitting leaves.
On the other hand, we can use causes for states not being distinguished to construct a distinguishing graph that distinguishes all or at least many of the not distinguished states.

Though our distinguishing graphs significantly improve the size of an $n$-complete test suite, the problem to compute good access sequences for such a test suite requires further research as well \cite{ncompletejournal}. Due to the output nondeterminism of suspension automata, we need an input-fairness assumption, to ensure that all outputs enabled from a state may eventually be observed. However, for access sequences we rather have a more adaptive strategy, in the spirit of \cite{mbtgames}, that reacts on the outputs as produced by the tested system rightaway.
%
%
Adaptively choosing access sequences means that for reaching the same state, different access sequences may be used. However, the proof of $n$-completeness of a test suite depends on using one unique access sequence for accessing the same state. It remains an open problem whether using different access sequences breaks $n$-completeness or not.

 \bibliographystyle{plain}
 \bibliography{lib}
 
\end{document}